\documentclass[square,numbers,superscriptaddress]{revtex4-2}
\usepackage[utf8]{inputenc}
\usepackage[margin=1in]{geometry}
\usepackage{amsmath}
\usepackage{amsfonts}
\usepackage{amssymb}
\usepackage{amsthm}
\usepackage{mathtools}
\usepackage{hyperref}
\usepackage{cleveref}
\usepackage[braket,qm]{qcircuit}

\usepackage{color}
\definecolor{Pr}{rgb}{0.4,0.3,0.9}

\usepackage{xcolor}

\DeclareMathOperator\erf{erf}
\DeclareMathOperator\erfs{erfs}
\DeclareMathOperator\sgn{sgn}
\DeclareMathOperator\hvy{hvy}

\usepackage{bm}

\newcommand{\lnorm}[1]{\left\lVert #1 \right\rVert}

\renewcommand{\op}[2]{|#1\rangle\langle#2|}
\renewcommand{\ip}[2]{\langle#1|#2\rangle}
\renewcommand{\bra}[1]{\langle#1|}
\renewcommand{\ket}[1]{|#1\rangle}

\renewcommand{\Re}[1]{\text{Re}(#1)}
\renewcommand{\Im}[1]{\text{Im}(#1)}

\def\be{\begin{eqnarray}}
\def\ee{\end{eqnarray}}

\newtheorem{lemma}{Lemma}
\newtheorem{theorem}{Theorem}
\newtheorem{remark}{Remark}

\newtheorem{definition}{Definition}

\DeclareMathOperator*{\argmax}{arg\,max}

\newcommand{\Z}{\hat{Z}}
\renewcommand{\H}{\hat{H}}
\renewcommand{\S}{\hat{S}}

\begin{document}

\title{Non-Linear Transformations of Quantum Amplitudes: Exponential Improvement, Generalization, and Applications}

\author{Arthur G. Rattew}
\email{arthur.rattew@materials.ox.ac.uk}
\affiliation{Department of Materials, University of Oxford, Parks Road, Oxford OX1 3PH, United Kingdom\looseness=-1}
\affiliation{Global Technology Applied Research, JPMorgan Chase, New York, NY 10017 USA}

\author{Patrick Rebentrost}
\email{cqtfpr@nus.edu.sg}
\affiliation{Centre for Quantum Technologies, National University of Singapore, Singapore 117543}

\begin{abstract}
Quantum algorithms manipulate the amplitudes of quantum states to find solutions to computational problems. In this work, we present a framework for applying a general class of non-linear functions to the amplitudes of quantum states, with up-to an exponential improvement over the previous work. Our framework accepts a state preparation unitary (or block-encoding), specified as a quantum circuit, defining an $N$-dimensional quantum state. We then construct a diagonal block-encoding of the amplitudes of the quantum state, building on and simplifying previous work. 
Techniques from the QSVT literature are then used to process this block-encoding. The source of our exponential speedup comes from the quantum analog of importance sampling. We then derive new error-bounds relevant for end-to-end applications, giving the error in terms of $\ell_2$-norm error. 
We demonstrate the power of this framework with four key applications.
First, our algorithm can apply the important function $\tanh(x)$ to the amplitudes of an arbitrary quantum state with at most an $\ell_2$-norm error of $\epsilon$, with worst-case query complexity of $O(\log(N/\epsilon))$, in comparison to the $O(\sqrt{N}\log(N/\epsilon))$ of prior work.
Second, we present an algorithm solving a new formulation of maximum finding in the unitary input model. 
Third, we prove efficient end-to-end complexities in applying a number of common non-linear functions to arbitrary quantum states. Finally, we generalize and unify existing quantum arithmetic-free state-preparation techniques.
Our work provides an important and efficient algorithmic building block with potentially numerous applications in areas such as optimization, state preparation, quantum chemistry, and machine learning.

\end{abstract}

\maketitle

\section {Introduction}

The macroscopic world often behaves in non-linear ways, and non-linear functions appear in many problems in fields such as engineering~\cite{ames1965nonlinear,finlayson2003nonlinear}, optimization~\cite{burer2012non}, machine learning~\cite{goodfellow2016deep,deng2014tutorial}, and finance~\cite{forsyth2007numerical,pistoia2021quantum}. Further illustrating their pervasiveness, the Komolgorov-Arnold theorem~\cite{kolmogorov1957representation} implies that all multivariate functions can be approximated by the summation and composition of univariate non-linear functions.
Clearly then, operating with non-linear functions is of significant and widespread importance. 
However, quantum algorithms find solutions to computational problems by manipulating the amplitudes of quantum states through unitary, and thus linear, transformations.  Therefore, it is a fundamental and important question how to make quantum algorithms exhibit non-linear behaviour.

Quantum signal processing~\cite{low2017optimal, low2019hamiltonian} and quantum singular value transformation \cite{gilyen2019quantum} are powerful frameworks encompassing nearly all of quantum computing~\cite{martyn2021grand}. 
They provide a machinery to develop quantum subroutines in a modular fashion, based on the block-encodings of arbitrary (not necessarily unitary) matrices in the subspaces of higher dimensional unitary operators. 
The power of these techniques originates from their ability to enact polynomial transformations on the eigenvalues and singular values of such encoded matrices, and so it is interesting to consider how these ideas apply to the non-linear transformation of quantum states. 
This naturally leads one to wonder: what are the best ways of performing such amplitude transformations, and what are the relevant applications? 

To answer this question, we now define an input model, and motivate the problem statement.
We assume we are given a state preparation unitary $U \in \mathbb C^{N \times N}$, $N=2^n$, specified by a quantum circuit, with $U \ket {0} = \sum_{j=1}^N \psi_j \ket j =: \ket{\psi}$, and we wish to enact non-linear transformations on the state amplitudes, potentially as a subroutine in another algorithm. We generalize this input model in \cref{section:state_preparation_block_encodings}. 

\begin{definition}[Non-Linear Amplitude Transformation, Informal]\label{problem:amplitude_transform}
    Let $U$ be an $n$-qubit state preparation unitary 
    specifying the $N = 2^n$ dimensional quantum state $\ket{\psi} = \sum_{j=0}^{N-1}\psi_j\ket{j}$. Let $f$ be a function $f:[-1,~1]~\mapsto~\mathbb R$, and let $\epsilon \ge 0$.
    The non-linear amplitude transformation problem is to prepare a state $\epsilon$-close (in some measure of distance) to the state, $\frac{1}{\mathcal{N}}\sum_{j = 0}^{N - 1} f(\psi_j)\ket{j}$ where $\mathcal{N}^2 := \sum_{j=0}^{N-1}|f(\psi_j)|^2$. 
\end{definition}

\subsection{Prior Work}

Some prior work exists in this context. 
Mitarai \textit{et al.} \cite{mitarai2019quantum} demonstrated how to convert the amplitudes of a quantum state into a register coupled to the state, thereby converting the ``analog'' amplitude into a ``digital'' representation. 
The digital representations of the state amplitudes can then be modified coherently via arithmetic operations, analogously to how the inverse is computed in the original quantum matrix inversion proposal~\cite{harrow2009quantum}, enabling a function to be applied to the discretized representation. Then, a controlled rotation circuit, as used implicitly in e.g., \cite{harrow2009quantum}, is applied to the digitized function values, converting them to amplitudes, with some probability of success. 
Using a subspace error definition, and excluding other complexity parameters, their query complexity to the state-preparation unitary is $O(1/\epsilon)$, analogous to phase estimation.

Their clever construction led \cite{guo2021nonlinear} to develop an algorithm which obtains an exponential improvement in the error-dependence for applying polynomial functions. Their construction leverages the ability of QSVT to implement non-linear transformations in controllable subspaces of unitary operators.
Their results also utilize a subspace error definition. In \Cref{appComparison}, and in particular in \cref{theorem:modified_variant_of_previous_technique_complexity_of_applying_a_polynomial_to_a_state}, we rederive their algorithm on polynomial transformations using our diagonal block-encoding variant, and with a new $\ell_2$-norm error guarantee. In \cref{theorem:generalized_end_to_end_complexity_for_previous_technique_for_efficiently_analytic_functions} we then generalize this result to implement transformations according to arbitrary functions which can be well-approximated by polynomials.  

The results of~\cite{guo2021nonlinear} are obtained by creating a (non-diagonal) block-encoding of the state amplitudes, implementing the non-linear polynomial transformation in the block-encoded subspace, and subsequently applying their block-encoding to a uniform superposition (conditioned on measuring a particular outcome in the ancillary register). However, by applying their block-encoding to a uniform superposition, they can incur a worst-case query complexity of $\tilde{O}(\sqrt{N})$, which is sub-optimal in certain important cases. As we later detail, the quantum analog of importance sampling can often be used to eliminate this dimension dependence. 

Solving computational problems with importance sampling can lead to dramatic improvements over uniform sampling. 
Classically, importance sampling is the basis for achieving exponential improvements in a variety of algorithms such as inner-product estimation, low-rank matrix inversion, and low-rank machine learning algorithms, often under the umbrella of ``quantum-inspired" algorithms~\cite{Tang2018,chia2022sampling,gilyen2022improved,bakshi2023improved}. In the quantum case, the power of having a quantum importance sampler was also demonstrated in quantum-accelerated Monte-Carlo algorithms \cite{Montanaro2015}, which avoid generally a dimension dependence in their complexity.

Finally, we additionally note recent work on non-linear amplitude transformations ~\cite{holmes2023nonlinear}.

\subsection{Our work}
In this paper, we build upon the works of Refs. \cite{mitarai2019quantum,guo2021nonlinear}. 
We now outline the document structure, and highlight our main contributions. 
In \Cref{section:prelims} we establish notation, common terminology, and necessary results. 
In \Cref{section:diagonal_block_encoding_of_state_amplitudes}, and in particular in \cref{theorem:diagonal_block_encoding_of_state_amplitudes}, we derive our diagonal variant of the block-encoding used in~\cite{guo2021nonlinear}, simplifying our subsequent proofs.
In \Cref{section:importance_weighted_amplitude_transform}, we derive our main result (\cref{theorem:polynomial_transformation_of_real_state_amplitudes_via_importance_sampling}) as well as a number of technical lemmas, enabling the polynomial transformation of the amplitudes of quantum states via importance sampling. This result is a major contribution of our paper, and is fundamental to many of the subsequent results. 
In essence, instead of applying the block-encoding of the state amplitudes (transformed by the desired polynomial) to a uniform superposition, as is done in the prior work, we instead apply a different polynomial to the state itself. As a consequence, we necessitate a number of technical results, and we ultimately remove the general dimension-dependence from the query and circuit complexities of the prior work for a broad class of functions. Precisely, given a polynomial $P(x)$ we instead work with polynomials $h(x)$ such that $x h(x) = P(x)$. In \Cref{section:non_linear_transformations_via_polynomial_approximations}, and in particular in \cref{theorem:generalized_end_to_end_complexity_for_f_0_is_0_for_efficiently_analytic_functions}, we generalize the result of the previous section to enable the transformation of state amplitudes by a broad class of functions which can be well approximated by such polynomials. These results are all derived with rigorous $\ell_2$-norm error guarantees. In \Cref{appComparison}, we derive analogous results for the algorithm in \cite{guo2021nonlinear}.

In \Cref{section:end_to_end_complexity_for_applying_various_functions}, we prove the efficient end-to-end complexity in applying non-linear transformations to the amplitudes of arbitrary quantum states for a broad class of functions, with the results presented in \cref{theorem:end_to_end_complexity_for_applying_various_functions}. In particular, we prove results for the exponential function, cosine function, logistic function, Gaussian distribution, and the sine function. In \Cref{section:exponential_improvement_tanh}, \cref{theorem:exponential_speedup_in_applying_tanh_demonstration}, we demonstrate the exponential speedup obtained by our importance-sampling based technique in applying the $\tanh(x)$ function when compared to the algorithm of \cite{guo2021nonlinear}.

In \Cref{section:maximum_finding}, we present a new maximum finding problem formulation, along with a corresponding algorithm in \cref{theorem:maximum_finding_in_state_preparation_unitary_input_model} derived utilizing our importance sampling results. We note that this problem format, while not identical, has some overlap with the $\ell_{\infty}$-norm tomography results of~\cite{van2021quantum, van2023quantum}, which warrants further consideration. 
Informally, give a state preparation unitary, this formulation of maximum finding asks to identify (with high probability of success) the basis state with maximum amplitude. 
Note that in terms of asymptotic complexity, it is an easier task to sample the state with maximum amplitude (while not being able to conclude that its the maximum) than it is to sample a state \textit{and} conclude that it has maximum amplitude. Where $\psi_{max}$ is the maximum amplitude of any basis vector, the former only requires $\tilde{O}(\psi_{max}^{-2})$ samples via direct sampling, while the latter necessarily has to have a dependence on the gap of the largest and second largest amplitudes. To this end, the algorithm we propose concludes that a given state has maximum amplitude by creating a mask (filtering out all states with amplitude smaller than the maximum in the constructed diagonal block encoding of the state amplitudes) and applying this mask to the input state. 
All this is automatically done using our importance-sampling results by invoking~\cref{theorem:generalized_end_to_end_complexity_for_f_0_is_0_for_efficiently_analytic_functions} with the appropriate non-linear function. 
These techniques can be readily extended to a number of applications e.g., filtering out basis vectors in a given state with low probability amplitude.

Finally, in \Cref{section:continuous_state_prep}, we demonstrate how the framework of non-linear transformation of quantum amplitudes generalizes and unifies the quantum-arithmetic-free state-preparation approaches of \cite{gonzalez2023efficient} and \cite{mcardle2022quantum}. Essentially, as observed by \cite{gonzalez2023efficient}, the problem of preparing a quantum state with amplitudes proportional to a given non-linear continuous function is naturally captured in the framework of non-linear amplitude transformations by simply applying the non-linear function to an easy-to-prepare quantum state with amplitudes proportional to the linear function (with some details regarding normalization).
The proof in \cref{theorem:state_preparation_through_non_linear_transformation} utilizes our diagonal block-encodings, as well as the results in \Cref{appComparison} and \Cref{section:state_preparation_block_encodings}.

\section{Preliminaries}\label{section:prelims}
We use $\widetilde O(\cdot)$ to hide a polylog factor, i.e., $\widetilde O(f(n)) = O(f(n)\cdot {\rm poly}\log(f(n)))$. We use the notation $[N]$ to represent the set of $N$ integers from $0$ to $N-1$, e.g. $\sum_{j\in [N]}$ represents a sum from $0$ to $N-1$. We define the natural numbers $\mathbb{N}$ to include $0$.
Define $\oplus$ as the bit-wise addition modulo $2$ of two bit strings (bit-wise XOR operation). 
Additionally, we use subscripts on kets to denote the number of qubits composing that state. For example, $\ket{\psi}_n$ represents some $n$ qubit quantum state; $\ket{\psi}_n \in \mathbb{C}^{2^n}$. When the dimension is clear from the context, we drop the subscript. 
For some function $f$ and a normal operator $H = \sum_j \lambda_j \op{\lambda_j}{\lambda_j}$, we use the notation $f(H) := \sum_j f(\lambda_j) \op{\lambda_j}{\lambda_j}$ to represent the transformation of the eigenvalues of $H$ by $f$. Following~\cite{gilyen2019quantum}, for a general operator $M$ with singular value decomposition $M = U \Sigma V^{\dagger}$, we use the notation $f^{(SV)}(M) := U f(\Sigma) V^{\dagger}$ to represent the transformation of the singular values if $f$ is an odd function. If $f$ is an even function, we instead use the notation $f^{(SV)}(M) := V f(\Sigma) V^{\dagger}$.

\textit{Computational model}-- We use the standard quantum circuit model of quantum computation~\cite{Nielsen2010}. An application of a quantum gate is equivalent to performing an elementary operation. The query complexity of a classical/quantum algorithm with some input is the maximum number of queries the algorithm makes on the input.
Time complexity is measured in circuit depth, in terms of single and two qubit gates. Classical time complexity is measured in the number of fundamental computation steps. We assume that all single and two qubit gates are implemented exactly.

\begin{definition}[Block encoding \cite{gilyen2019quantum}]\label{def:quantum_block_encoding}
Suppose that $A$ is an $s$-qubit operator, $\alpha,\epsilon \in\mathbb R_+$ and $a\in \mathbb N$, then we say that the $(s+a)$-qubit unitary $U$ is an $(\alpha,a,\epsilon)$-block-encoding of $A$, if 
$$ \Vert A - \alpha(\bra{0}^{\otimes a}\otimes I)U(\ket{0}^{\otimes a}\otimes I)\Vert \leq \epsilon. $$
\end{definition}

\begin{definition}[Matrix eigenvalue functions]\label{def:matrix_eigenvalue_function}
Given a $N\times N$ normal matrix $H = \sum_{j=1}^N \lambda_j\op{\lambda_j}{\lambda_j}$, with $H\ket{\lambda_j}= \lambda_j \ket{\lambda_j}$, and a function $f : \mathbb{C} \mapsto \mathbb{C}$, we define a matrix function on its eigenvalues, i.e. $f(H) := \sum_{j=1}^N f(\lambda_j)\op{\lambda_j}{\lambda_j}$.
\end{definition}

\begin{definition}[Matrix singular value functions]
Given a function $f : [-1, 1] \mapsto \mathbb{C}$ and a matrix $H$ with singular value decomposition $H = U \Sigma V^{\dagger}$, we define a singular value function as $f^{(SV)}(H) := U f(\Sigma) V^{\dagger}$, where the function on $\Sigma$ is defined as per \cref{def:matrix_eigenvalue_function}.
\end{definition}

\begin{definition}[Lipschitz constant and Lipschitz continuity]\label{def:lipschitz_constant}
If a function $f$ is Lipschitz continuous on the interval $[a, b]$, there exists some constant $L \ge 0$ such that for all $x\neq y$ (with $x \in [a,b]$ and $y\in[a,b]$),
\begin{align}
    \frac{|f(y) - f(x)|}{|y - x|} \le L.
\end{align}
Moreover, we call the smallest $L$ satisfying the preceding equation the Lipschitz constant; the Lipschitz constant $L$ is given by $L := \max_{x\in[a, b]}|f'(x)|$.
\end{definition}

We begin by stating the result enabling the quantum eigenvalue transform of Hermitian matrices by arbitrary complex polynomials of \cite{gilyen2019quantum}.

\begin{theorem}[Theorem 56 of \cite{gilyen2019quantum}]\label{theorem:qet_of_hermitian_matrices}
Given a degree$-k$ polynomial of the form $P(x) = \sum_{j=0}^k a_j x^j$ with $x \in [-1, 1]$ and $\forall j, a_j \in \mathbb{C}$ such that $\forall x, |P(x)| \le 1/4$, a $(\alpha, a, \epsilon)$-block-encoding $U_H$ of a Hermitian matrix $H$, one can construct a quantum circuit $\tilde{U}$, implementing a $(1, a + 2, 4k\sqrt{\epsilon/\alpha} + \delta)$-block-encoding of $P(H/\alpha)$. $\tilde{U}$ can be constructed with $O(ka)$ single and two-qubit gates, and a total of $k$ calls to a controlled $U_H$ and controlled $U_H^{\dagger}$ circuit. Moreover, the circuit description of $\tilde{U}$ can be computed with classical time complexity of $O(\text{poly}(k, \log(1/\delta))$.
\end{theorem}

\begin{lemma}[Product of block encodings \cite{gilyen2019quantum}]\label{lemma:product_of_block_encodings}
If $U$ is an $(\alpha,a,\delta)$-block-encoding of an $s$-qubit operator $A$, and $V$ is an $(\beta,b,\epsilon)$-block-encoding of an $s$-qubit operator $B$ then $(I_b\otimes U)(I_a\otimes V)$ is an $(\alpha\beta,a+b,\alpha\epsilon+\beta\delta)$-block-encoding of $AB$. 
\end{lemma}
In \cref{lemma:product_of_block_encodings}, following the convention of \cite{gilyen2019quantum}, it is assumed that $U$ acts trivially on the ancillas of $V$, and $V$ acts trivially on the ancillas of $U$ -- this notation is not used anywhere else in this paper. I.e., the tensor products in $(I_b\otimes U)(I_a\otimes V)$ are not read in the usual sense in this lemma.

\begin{lemma}\label{lemma:normalized_state_deviation_due_to_normalization}
    Given two normalized quantum states $\ket{\tilde{a}} = \frac{1}{\mathcal{N}_a}\ket{a}$, and $\ket{\tilde{b}} = \frac{1}{\mathcal{N}_b}\ket{b}$ such that $|\mathcal{N}_a - \mathcal{N}_b| \le \epsilon_0$, and that $\lnorm{\ket{a} - \ket{b}}_2 \le \epsilon_1$, then $\lnorm{\ket{\tilde{a}} - \ket{\tilde{b}}}_2 \le (\epsilon_0 + \epsilon_1)/\max\{\mathcal{N}_a,\mathcal{N}_b \}$.
\end{lemma}
\begin{proof}
A simple calculation obtains 
$\lnorm{\ket{\tilde{a}} - \ket{\tilde{b}}}_2
=
\lnorm{\frac{\ket{a}}{\mathcal{N}_a} - \frac{\ket{b}}{\mathcal{N}_b}}_2
=
\frac{1}{\mathcal{N}_a\mathcal{N}_b}
\lnorm{(\mathcal{N}_b - \mathcal{N}_a)\ket{a} + \mathcal{N}_a(\ket{a} - \ket{b})}_2$.
Applying triangle inequality, and noting that $\mathcal{N}_a = \lnorm{\ket{a}}_2$,
$\lnorm{\ket{\tilde{a}} - \ket{\tilde{b}}}_2
\le 
\frac{1}{\mathcal{N}_b}
\left( 
|\mathcal{N}_b-\mathcal{N}_a|
+
\lnorm{\ket{a} - \ket{b}}_2
\right)
\le 
\frac{\epsilon_0 + \epsilon_1}{\mathcal{N}_b}$.
\end{proof}

\section{Diagonal block-encoding of quantum state amplitudes}\label{section:diagonal_block_encoding_of_state_amplitudes}
In \cref{def:state_prep_unitary_input}, we begin by specifying the assumed input model, a state preparation unitary as specified by a quantum circuit. When applied to the $\ket{0}$ state, this quantum circuit uniquely specifies a quantum state. We then present some further preliminary definitions and results, culminating in the derivation of our diagonal block-encoding of the state amplitudes, a variant of the block-encoding first presented in~\cite{mitarai2019quantum,guo2021nonlinear}.
We generalize this input model in the appendix. 
 
\begin{definition}[Input model of this work, State preparation unitary]\label{def:state_prep_unitary_input}
Let $n \in \mathbb Z_{> 0}$.
We are given an $n$-qubit quantum circuit termed the ``state-preparation circuit''. This circuit is described by a unitary matrix denoted by $U \in \mathbb C^{N\times N}$,  where $N=2^n$.
The action of $U$ on the all-zero state defines the following state,
\begin{align}
U\ket{0}_n = \sum_{j=0}^{N-1} \psi_j \ket{j} =: \ket{\psi},
\end{align}
where $\psi_j \in \mathbb C$ with $\sum_{j=0}^{N-1} \vert\psi_j\vert^2=1$. 
We call a ``real state preparation circuit'' a quantum circuit that has the property that $\forall j, \psi_j \in \mathbb R$.
\end{definition}

We define a few elementary gates in the following. We use the total working space of $2n+1$ qubits, in anticipation of the needs of the algorithm. The definitions make precise on which qubits the gates are acting on.
\begin{definition}\label{def:R_gate}
We define the following rotation gate, acting on $2n + 1$ qubits:
\begin{align}
    R := (I_{n+1} - 2\op{0}{0}_{n+1})\otimes I_n.
\end{align}
\end{definition}
This gate can be implemented by an $n$-controlled $XZX$ gate targeting the first qubit, controlled on the next $n$-qubits (controlled on the $\ket{0}_n$ state), and with an identity acting on the final $n$ qubits. The $n+1$ qubit controlled gate can then be decomposed into $O(n)$ two-qubit gates via the decomposition of \cite{saeedi2013linear}, giving a total circuit depth for this operator of $O(n)$.

\begin{definition}\label{def:single_qubit_defs}
We define the following $2n + 1$ qubit gates,
\begin{align}
\Z := I_n \otimes Z \otimes I_n,\quad   \H := I_n \otimes H \otimes I_n,\quad
\S := I_n \otimes S \otimes I_n,
\end{align}
where $Z$ is the standard $1$-qubit Pauli-Z gates, $H$ is the $1$-qubit Hadamard gate, and $S = \begin{pmatrix}1 & 0\\ 0 & i \end{pmatrix}$. 
\end{definition}
These can all clearly be implemented with  a single circuit layer.

\begin{definition}\label{def:U_C_circuit_def}
Given an $n$-qubit unitary $U$, we define the $2n+1$ qubit operator $U_C$ as follows
\begin{align}
    U_C := (U\otimes \op{0}{0}_1 + I_n \otimes \op{1}{1}_1)\otimes I_n.
\end{align}
\end{definition}
This gate can be implemented with one query to a controlled $U$ gate, with an $X$ gate applied before and after the control (on the control qubit).

\begin{definition}\label{def:copy_circuit}
We define a $2n + 1$ qubit ``controlled-copy'' circuit $C$, which implements the mapping of the $\ket{0}_n\ket{0}_1\ket{k}_n$ state to itself, and the $\ket{0}_n\ket{1}_1\ket{k}_n$ state to $\ket{k}_n\ket{1}_1\ket{k}_n$. The matrix form of the gate may be written as,
\begin{align}
    C := I_{n}\otimes \op{0}{0}_1 \otimes I_n
    +
    \sum_{k=0}^{N-1}\sum_{j=0}^{N-1}\op{j\oplus k}{j}_n\otimes \op{1}{1}_1 \otimes \op{k}{k}_n.
\end{align}
\end{definition}
This gate can be implemented by $n$ Toffoli gates controlled on the $\ket{1}_1$ state of the middle qubit, and with the other control of the $i^{th}$ Toffoli gate acting on the $\ket{1}_1$ state of the $(n + 1 + i)^{th}$ qubit, and with the target being the $i^{th}$ qubit. This can clearly be implemented with $O(n)$ two-qubit gates and $O(n)$ circuit depth. 

We continue by defining some gate sequences and proving their action. We define the $W_1$ operator slightly differently than in \cite{mitarai2019quantum, guo2021nonlinear}, but note that the action on the desired subspace is unchanged.
The flag $p$ determined if a relative phase is applied or not. 
\begin{lemma}[$W$ operator from \cite{mitarai2019quantum,guo2021nonlinear}]\label{lemma:W_operator_definition_and_action}
For $p\in\{0,1\}$, define the $2n + 1$ qubit operator $W_p$ by the circuit,
\begin{align}
    W_p := \H \hat S^p C U_C \H.
\end{align}
When operating on the states $\ket{0}_n\ket{0}_1\ket{k}_n$, where $\ket{k}_n$ is an $n$-qubit standard basis vector, it prepares the family of $2^n$ states $\{\ket{\Phi_k^p}\}_k$ given by
\begin{align}
\ket{\Phi_k^p} := W_p\ket{0}_n\ket{0}_1\ket{k}_n
=
\frac{1}{2}\left(
(\ket{\psi}_n + i^p \ket{k}_n)\ket{0}_1 + (\ket{\psi}_n - i^p \ket{k}_n)\ket{1}_1
\right)\ket{k}_n.
\end{align}
The circuit $W_p$ uses $O(n)$ circuit depth with $1$ query to a controlled-$U$ gate. 
\end{lemma}
\begin{proof}
First, with $\H$ as per \cref{def:single_qubit_defs},
$\H \ket{0}_n\ket{0}_1\ket{k}_n
=
\frac{1}{\sqrt{2}}\left(
\ket{0}_n\ket{0}_1 + \ket{0}_n\ket{1}_1
\right)\ket{k}_n$.
Then, with $U_C$ as per \cref{def:U_C_circuit_def}, since $U\ket{0}_n = \ket{\psi}_n$,
$
U_C\H \ket{0}_n\ket{0}_1\ket{k}_n
=
\frac{1}{\sqrt{2}}\left(
\ket{\psi}_n\ket{0}_1 + \ket{0}_n\ket{1}_1
\right)\ket{k}_n$.
Applying $C$ as per \cref{def:copy_circuit} then yields,
\begin{align}\label{eqn:W_operator_common_part_with_W_i}
C U_C\H \ket{0}_n\ket{0}_1\ket{k}_n
    =
    \frac{1}{\sqrt{2}}\left(
    \ket{\psi}_n\ket{0}_1 + \ket{k}_n\ket{1}_1
    \right)\ket{k}_n.
\end{align}
As $\S^0 = I_{2n+1}$, applying $\S^p$ (as per \cref{def:single_qubit_defs}) then yields,
\begin{align}
    \S^p C U_C\H \ket{0}_n\ket{0}_1\ket{k}_n
    =
    \frac{1}{\sqrt{2}}\left(
    \ket{\psi}_n\ket{0}_1 + i^p\ket{k}_n\ket{1}_1
    \right)\ket{k}_n.
\end{align}
Finally, applying $\H$ again yields,
\begin{align}
W_p\ket{0}_n\ket{0}_1\ket{k}_n
    =
    \frac{1}{2}\left(
    (\ket{\psi}_n + i^p \ket{k}_n)\ket{0}_1 + (\ket{\psi}_n - i^p\ket{k}_n)\ket{1}_1
    \right)\ket{k}_n.
\end{align}
Clearly, we have implemented this transformation with $O(n)$ circuit depth, and one query to a controlled-$U$ gate.
\end{proof}

\begin{lemma}\label{lemma:eigenvectors_and_eigenvalues_of_G_plus_G_dagger}
Let $p \in \{0,1\}$. Let $G_p:=W_p R W_p^{\dagger}\Z$ as per~\cite{guo2021nonlinear}. Then,
$\ket{\Phi_k^p}$ is an eigenstate of $-\frac{1}{2}(G_p + G_p^{\dagger})$, with eigenvalue $\Re{\psi_k}$, if $p=0$, and $\Im{\psi_k}$, if $p=1$.
\end{lemma}
\begin{proof}
This fact was first identified and proven in its current form in \cite{guo2021nonlinear}, and in a similar form in \cite{mitarai2019quantum}.
First, since $\Z = \Z^{\dagger}$, $R = R^{\dagger}$, we have that 
$G_p^{\dagger} = \Z W_p R W_p^{\dagger}.$
In addition, with $R = (I_{n+1} - 2\op{0}{0}_{n+1})\otimes I_n$, we have
\begin{align}
G_p^{\dagger}\ket{\Phi_k^p} = \Z W_p R W_p^{\dagger}\ket{\Phi_k^p}
=
\Z W_p R \ket{0}_n\ket{0}_1\ket{k}_n
=
- \Z W_p \ket{0}_n\ket{0}_1\ket{k}_n
=
-\Z \ket{\Phi_k^p}.
\end{align}
As a result, we have
\begin{align}
    (G_p + G_p^{\dagger})\ket{\Phi_k^p}
    =
    W_p R W_p^{\dagger}\Z\ket{\Phi_k^p}
    -\Z \ket{\Phi_k^p}
    = 
    (W_p R W_p^{\dagger} - I_{2n + 1})\Z \ket{\Phi_k^p}.
\end{align}
Using the definition of $R$, we obtain 
\begin{align}
    W_p R W_p^{\dagger} = W_p W_p^{\dagger} - 2W_p (\op{0}{0}_{n+1}\otimes I_n)W_p^{\dagger}
    =
    I_{2n+1} - 2W_p (\op{0}{0}_{n+1}\otimes I_n)W_p^{\dagger}.
\end{align}
Thus,
\begin{align}\label{eqn:proving_eigenvector_of_G_G_dagger_step_2}
    -\frac{1}{2}(G_p + G_p^{\dagger})\ket{\Phi_k^p}
    =
     W_p (\op{0}{0}_{n+1}\otimes I_n) W_p^{\dagger} \Z \ket{\Phi_k^p}.
\end{align}
We will now examine $W_p^{\dagger}\Z \ket{\Phi_k^p}$, noting that 
\begin{align}
    \Z\ket{\Phi_k^p} = \frac{1}{2}\left((\ket{\psi}_n +i^p \ket{k}_n)\ket{0}_1 -(\ket{\psi}_n -i^p\ket{k}_n)\ket{1}_1 \right)\ket{k}_n.
\end{align}
Since $W_p= \H \S^p C U_C \H$, and thus $W_p^{\dagger} = \H U_C^{\dagger}C^{\dagger}(\S^p)^\dagger\H$. Applying $W_p^{\dagger}$ one part at a time, we first obtain
\be
    (\S^p)^\dagger \H \Z\ket{\Phi_k^p} 
    &=&
    (\S^p)^\dagger \frac{1}{2}\left(
    (\ket{\psi}_n + i^p\ket{k}_n)\ket{+}_1 - (\ket{\psi}_n - i^p \ket{k}_n)\ket{-}_1
    \right)\ket{k}_n \\
    &=&
    \frac{1}{\sqrt{2}}
    \left(
    i^p \ket{k}_n\ket{0}_1 + (-i)^p \ket{\psi}_n\ket{1}_1
    \right)\ket{k}_n = \frac{i^p}{\sqrt{2}}
    \left(
    \ket{k}_n\ket{0}_1 + (-1)^p \ket{\psi}_n\ket{1}_1
    \right)\ket{k}_n,
\ee
Applying $C^{\dagger} = I_{n}\otimes \op{0}{0}_1 \otimes I_n + \sum_{k=0}^{N-1}\sum_{j=0}^{N-1}\op{j}{j\oplus k}_n\otimes \op{1}{1}_1 \otimes \op{k}{k}_n$, then yields
\begin{align}
C^{\dagger}(\S^p)^\dagger\H \Z\ket{\Phi_k^1}
=
\frac{i^p}{\sqrt{2}}\left(
    \ket{k}_n\ket{0}_1 + (-1)^p\sum_{j=0}^{N-1}\psi_{j \oplus k}\ket{j}_n\ket{1}_1
\right)\ket{k}_n.
\end{align}
Applying $(U^{\dagger} \otimes \op{0}{0}_1\otimes I_n + I_n \otimes \op{1}{1}_1\otimes I_n)$ then gives,
\begin{align}
    (U^{\dagger} \otimes \op{0}{0}_1\otimes I_n + I_n \otimes \op{1}{1}_1\otimes I_n) C^{\dagger} \H \Z\ket{\Phi_k^p}
    =
    \frac{i^p}{\sqrt{2}}\left(
        U^{\dagger}\ket{k}_n\ket{0}_1 +(-1)^p \sum_{j=0}^{N-1}\psi_{j \oplus k}\ket{j}_n\ket{1}_1
    \right)\ket{k}_n.
\end{align}
Finally, applying $\H$ gives,
\begin{align}
W_p^{\dagger}\Z \ket{\Phi_k^p}
&=
\frac{i^p}{\sqrt{2}}\left(
U^{\dagger}\ket{k}_n\ket{+}_1 + (-1)^p\sum_{j=0}^{N-1}\psi_{j \oplus k}\ket{j}_n\ket{-}_1\right)\ket{k}_n.
\end{align}
Then, $(\op{0}{0}_{n+1}\otimes I_n) W_p^{\dagger} \Z \ket{\Phi_k^p}$ becomes: 
\begin{align}
    (\op{0}{0}_{n+1}\otimes I_n) W_p^{\dagger} \Z \ket{\Phi_k^p}
    &=
    \frac{i^p}{\sqrt{2}}(\op{0}{0}_{n+1}\otimes I_n)\left(
U^{\dagger}\ket{k}_n\ket{+}_1 +(-1)^p \sum_{j=0}^{N-1}\psi_{j \oplus k}\ket{j}_n\ket{-}_1\right)\ket{k}_n\\
    &=
    \frac{i^p}{2}\left(
        \ket{0}_n\ip{\psi}{k}\ket{0} + (-1)^p \psi_k\ket{0}_n\ket{0}_1
    \right)\ket{k}_n\\
    &=
    \frac{i^p}{2}(\psi_k^* + (-1)^p \psi_k)\ket{0}_n\ket{0}_1\ket{k}_n.
\end{align}
As a result, from \cref{eqn:proving_eigenvector_of_G_G_dagger_step_2}, we obtain 
\be
-\frac{1}{2}(G_p + G_p^{\dagger}) \ket{\Phi_k^p} &=& W_p (\op{0}{0}_{n+1}\otimes I_n) W_p^{\dagger} \Z \ket{\Phi_k^p}
=
\frac{i^p}{2}(\psi_k^* + (-1)^p \psi_k) \ket{\Phi_k^p}\\
&=&
 \begin{cases}\text{Re}(\psi_k) \ket{\Phi_k^p} \ \text{if}\ p=0. \\
\text{Im}(\psi_k) \ket{\Phi_k^p} \ \text{if}\ p=1.
\end{cases}.
\ee
\end{proof}

\begin{lemma}\label{lemma:basis_vectors_of_G_plus_G_dagger_diagonalized}
Let $p=\{0,1\}$. The operator $-\frac{1}{2}W_p^{\dagger}(G_p + G_p^{\dagger})W_p$ has eigenvectors $\ket{0}_n\ket{0}_1\ket{k}_n$ with associated eigenvalue $\Re{\psi_k}$ if $p=0$ and 
$\Im{\psi_k}$ if $p=1$. 
\end{lemma}
\begin{proof}
This simply follows from $W_p \ket{0}_n\ket{0}_1\ket{k}_n = \ket{\Phi_k^p}$, $W_p^\dagger\ket{\Phi_k^p} = \ket{0}_n\ket{0}_1\ket{k}_n$,
and \cref{lemma:eigenvectors_and_eigenvalues_of_G_plus_G_dagger}.
\end{proof}

\begin{lemma}\label{lemma:diagonal_block_encoding_of_real_state_parts}
Let $p=\{0,1\}$. The unitary 
\begin{align}
U^{(p)}:= (XZX \otimes I_{2n + 1})(H\otimes W_p^{\dagger})(\op{0}{0}_1\otimes G_p + \op{1}{1}_1\otimes G_p^{\dagger})(H\otimes W_p).
\end{align}
is a $(1, n + 2, 0)$-block-encoding of $A^{(p)}:=
\text{diag}(\Re{(-i)^p\psi_1}, ..., \Re{(-i)^p\psi_N})$, which is a shorthand for 
$\text{diag}(\Re{\psi_1}, ..., \Re{\psi_N})$ if $p=0$ and $\text{diag}(\Im{\psi_1}, ..., \Im{\psi_N})$ if $p=1$. The circuit $U^{(p)}$ may be implemented with $O(n)$ circuit depth, with a total of $6$ queries to a controlled-$U$ circuit. 
\end{lemma}
\begin{proof}
We aim to show that 
$
\lnorm{A^{(p)} - (\bra{0}^{\otimes n + 2}\otimes I_n)U^{(p)}(\ket{0}^{\otimes n + 2}\otimes I_n)}_2 = 0$.
We begin by simplifying the expression of $U^{(p)}$,
\begin{align}
U^{(p)}
&=(XZX\op{+}{+}_1)\otimes (W_p^{\dagger}G_p W_p) + (XZX\op{-}{-}_1)\otimes (W_p^{\dagger}G_p^{\dagger} W_p).
\end{align}
As a result,
\begin{align}\label{eqn:block_encoding_real_part_mid_proof}
&(\bra{0}^{\otimes n + 2}\otimes I_n)U^{(p)}(\ket{0}^{\otimes n + 2}\otimes I_n)\\
&=-
(\bra{0}_1\bra{0}^{\otimes n + 1}\otimes I_n)\left(
(\op{-}{+}_1)\otimes (W_p^{\dagger}G_p W_p) + \op{+}{-}_1)\otimes (W_p^{\dagger}G_p^{\dagger} W_p)
\right)(\ket{0}_1\ket{0}^{\otimes n + 1}\otimes I_n)\\
&=
-\frac{1}{2}
(\bra{0}^{\otimes n + 1}\otimes I_n)
W_p^{\dagger}(G_p + G_p^{\dagger})W_p
(\ket{0}^{\otimes n + 1}\otimes I_n).
\end{align}
Consider the $\bra{j}_n \cdot \ket{k}_n$ matrix element of this matrix,
$
-\frac{1}{2}
\bra{0}_n\bra{0}_1\bra{j}_n
W_p^{\dagger}(G_p + G_p^{\dagger})W_p
\ket{0}_n\ket{0}_1\ket{k}_n$.
From \cref{lemma:basis_vectors_of_G_plus_G_dagger_diagonalized}, we know that
$W_p^{\dagger}(G_p + G_p^{\dagger})W_p
\ket{0}_n\ket{0}_1\ket{k}_n = -2 \Re{(-i)^p\psi_k} \ket{0}_n\ket{0}_1\ket{k}_n$,
and so we prove the assertion.

Clearly, the circuit cost of $U^{(p)}$ is dominated by $W_p$, $G_p$ and their inverses. As per \cref{lemma:W_operator_definition_and_action}, each $W_p$ and $W_p^{\dagger}$ can be implemented with one query to the controlled-$U$ gate (and its inverse), and $O(n)$ additional depth in terms of single and two qubit gates. The $G_p$ operator can be implemented with a call to $W_p$ and $W_p^{\dagger}$ (adding two additional queries to a controlled-$U$ gate, and $O(n)$ additional circuit depth), a $\Z$ gate (with $O(1)$ depth) and the $R$ gate which has $O(n)$ depth as per \cref{def:R_gate}, for a total of one query to the controlled-$U$ gate (and one query to its inverse) and $O(n)$ depth. Thus, $U^{(p)}$ can be implemented with $O(n)$ circuit depth, 3 queries to a controlled-$U$ gate and $3$ queries to an inverse controlled-$U$ gate.
\end{proof}

\begin{theorem}[Diagonal block encoding of  amplitudes]\label{theorem:diagonal_block_encoding_of_state_amplitudes}
Given an $n$-qubit quantum state specified by a state-preparation-unitary $U$, such that $\ket{\psi}_n = U\ket{0}_n = \sum_{j=0}^{N-1}\psi_j\ket{j}_n$ (with $\psi_j\in\mathbb{C}$), we can prepare a $(1, n + 3, 0)$-block-encoding $U_A$ of the diagonal matrix $A=\text{diag}(\psi_0, ..., \psi_{N-1})$ with $O(n)$ circuit depth and a total of $O(1)$ queries to a controlled-$U$ gate.

\end{theorem}
\begin{proof}
From \cref{lemma:diagonal_block_encoding_of_real_state_parts} we obtain the circuits $U^{(p)}$, which block-encode the matrices $A^{(p)}$, the diagonal matrices of the real and imaginary part of the $\psi_j$s, respectively. Using the linear combination theorem for block-encoded matrices \cite{gilyen2019quantum}, we can construct a block encoding of  $A^{(0)} + i A^{(1)}$ with the stated number of ancillas, circuit depth and queries. The $\frac{1}{2}$ factor is removed with oblivious amplitude amplification (e.g. \cite{gilyen2019quantum, quek2021fast}), leaving the asymptotic cost of the circuit unchanged.
\end{proof}

\begin{lemma}[Polynomial transformation of real/imaginary part of amplitudes]\label{lemma:polynomial_transformation_of_the_real_part_block_encoding}
Let $\delta >0$. Given a $(1, n+2, 0)$-block-encoding $U^{(p)}$ of $A^{(p)}=\text{diag}(\Re{(-i)^p\psi_0}, ..., \Re{(-i)^p\psi_{N-1}})$ (with $U\ket{0}_n = \sum_j \psi_j\ket{j}_n$), and a degree-$k$ polynomial $P(x) = \sum_{j=0}^k a_j x^j$ (such that $\forall j, a_j\in \mathbb{C}$, and $\forall x, |P(x)| \le 1/4$) we can obtain a $(1, n + 4, \delta)$-block-encoding $\tilde{U}_P^{(p)}$ of $P(A^{(p)})$, which can be constructed with $O(kn)$ circuit depth, and a total of $k$ calls to the controlled $U^{(p)}$ and $(U^{(p)})^\dagger$ circuits. The circuit of $\tilde{U}_P^{(p)}$ can be computed with classical time complexity of $O(\text{poly}(k, \log(1/\delta))$. 
\end{lemma}
\begin{proof}
The resulting block encoding is very similar to the construction in \cite{guo2021nonlinear}, only it uses our diagonal block encoding (as opposed to the block encoding in \cite{guo2021nonlinear} which is not diagonal in the standard basis). The proof follows trivially by applying \cref{theorem:qet_of_hermitian_matrices} to the block encoding of $A^{(p)}$ constructed in \cref{lemma:diagonal_block_encoding_of_real_state_parts}, noting that $A^{(p)}$ is Hermitian.
\end{proof}

\section{Importance-weighted amplitude transform}\label{section:importance_weighted_amplitude_transform}
In this section, we derive the technique responsible for our exponential improvement over prior work.  In essence, rather than applying the desired polynomial to a block-encoding of the state amplitudes and then applying that to the uniform superposition (as is done in~\cite{guo2021nonlinear}), we instead consider a different polynomial transformation and then apply that to the state $U\ket{0}$. As a result, we first consider a technical lemma about the Lipschitz constant of such functions. 

\begin{lemma}\label{lemma:upper_bound_on_max_of_reduced_polynomial}
Let $f(x)$ be a real-valued univariate function such that $f(0) = 0$ with Lipschitz constant $L$ in the interval $x\in[-1, 1]$, (i.e. $L := \max_{x\in[-1,1]}|\frac{d}{dx}P(x)|$). Let $h(x)$ be a function such that $f(x) = xh(x)$. Then, 
\begin{align}
    \max_{x\in [-1, 1]} |h(x)| \le L.
\end{align}
\end{lemma}

\begin{proof}
First, taking the derivative with respect to $x$ of $f(x) = xh(x)$, we get $h(x) = f'(x) - xh'(x)$, i.e., $\max_x |h(x)| = \max_x |f'(x) - xh'(x)|$. 
Let $x_* := \argmax_{x\in[-1, 1]}|h(x)|$. If $|x_*| < 1$, the maximum of $|h(x)|$ must be a local optima, and so $h'(x_*) = 0$. Consequently, we get the bound $|h(x_*)| = L$. Note that while $h(x)$ is both continuous and differentiable, $|h(x)|$ is only continuous but not necessarily differentiable. However, $|h(x)|$ is only non-differentiable at points where the function is crossing the $x$-axis, which by definition cannot be at a local optima in the interval $x\in(-1, 1)$. 
When $|x_*| = 1$, $|h(x_*)| = |f(x_*)/x_*| = |f(x_*)| \le \gamma$.
We will now show that $\gamma \le L$, thereby proving the upper-bound of $L$ in the entire interval $x\in[-1, 1]$.

Redefine $x_* := \argmax_{x\in[-1,1]}|f(x)|$. 
The mean value theorem says that for some continuous differentiable function $f(x)$ on the interval $[a, b]$,  there exists some $c$ such that $f'(c) = (f(b) - f(a))/(b-a)$. We break our proof into two cases, when $x_* \ge 0$, and when $x_* < 0$.
When $x_* \ge 0$, there exists some interval $[0, x_*]$ such that for some $c$, $P'(c) = P(x_*)/x_*$ (where we used the fact that $P(0) = 0$). Since $P(x) = xh(x)$, $P(x)/x$ will never be singular. Consequently, we have that
\begin{align}
\gamma = |P(x_*)| \le |P(x_*)/x_*| = |P'(c)| \le \max_{x\in[-1, 1]}|P'(x)| = L.
\end{align}
Similarly, when $P(x_*) < 0$, there exists some interval $[x_*, 0]$ such that for some $c$, $P'(c) = P(x_*)/x_*$, and the preceding argument once again applies, proving that $\gamma \le L$. 
As a result, since we have shown that $\gamma \le L$, we have the overall bound $\max_{x\in[-1, 1]}|h(x)| \le L$.
\end{proof}

The next lemma provides simple way to obtain a Lipschitz constant from the knowledge of the coefficients of a polynomial. In addition, a bound on the maximum of $h(x) = P(x)/x$ is obtained. 
\begin{lemma}\label{lemma:upper_bound_on_max_of_reduced_polynomial2}
Let $c \in \mathbb R^{k+1}$ with $c_0=0$ and let $P(x) = \sum_{j=0}^k c_j x^j$ be a degree-$k$ polynomial function.
Let $h$ be the degree-$(k-1)$ polynomial function such that $P(x) = xh(x)$.
Then,
\begin{itemize}
\item For $x\in[-1, 1]$, a valid Lipschitz constant for $P$ is  $L:= k \Vert c\Vert_1$.
\item $
    \max_{x\in [-1, 1]} |h(x)| \le L/k  $.
\end{itemize}
\end{lemma}
\begin{proof}
Any constant $L$ for which $ \max_{x\in[-1,1]}|\frac{d}{dx}P(x)| \leq L$  is a valid Lipschitz constant on the interval $x\in[-1,1]$.
Here, $\frac{d}{dx}P(x) = \sum_{j=1}^k j c_j x^{j-1} \leq k \Vert c\Vert_1 =:L$ as $x\in[-1, 1]$. By definition, $h(x) = \sum_{j=1}^k c_j x^{j-1}$, and $\vert h(x)\vert \leq \Vert c\Vert_1\leq L/k$.
\end{proof}

\begin{theorem}[Polynomial transformation of real amplitudes via importance sampling]\label{theorem:polynomial_transformation_of_real_state_amplitudes_via_importance_sampling}
Let $\epsilon \in(0,1]$. Assume we are given a state preparation unitary $U\ket{0}_n = \ket{\psi}_n = \sum_{j=1}^{2^n} \psi_j\ket{j}_n$, where $\forall j, \psi_j \in \mathbb{R}$ with $\sum_{j=1}^{2^n} \vert\psi_j\vert^2$.
We are given also a degree $k$ polynomial function $P(x) = \sum_{j=1}^k a_j x^j$ ($\forall j\in[k+1], a_j \in \mathbb{C}$) with Lipschitz constant $L$ (in the interval $x\in[-1, 1]$). 
Define $\eta := \max_{x\in[-1, 1]}|P(x)/x|$, and note that \cref{lemma:upper_bound_on_max_of_reduced_polynomial} implies that $\eta \le L$.
Define $\mathcal{N}^2 := \sum_{j=1}^{2^n}|P(\psi_j)|^2$.  
Then, we can construct an $\ell_2$-normalized quantum state $\ket{\phi}$ such that $\lnorm{\ket{\phi} - \frac{1}{\mathcal{N}}\sum_{j=1}^{2^n} P(\psi_j)\ket{j}}_2 \le \epsilon$ with arbitrarily high probability of success with a total of $O(k\eta/\mathcal{N})$ queries to a controlled-$U$ circuit and a controlled-$U^{\dagger}$ circuit, with $O(kn\eta/\mathcal{N})$ additional circuit depth, and with a total of $O(poly(k, \log(\eta^2 N/\epsilon^2\mathcal{N}^2)))$ classical computation. This requires $O(n)$ ancillary qubits.
\end{theorem}
\begin{proof}
Let $A:=\text{diag}(\psi_0, ..., \psi_{N-1})$.
The degree $k-1$ polynomial function $h(x)$ such that $P(x) = h(x) x$ is given by $h(x) := \sum_{j=1}^{k}a_{j}x^{j-1}$. Let $\eta := \max_{x\in[-1,1]} |h(x)|$. Then, we define another function, $g(x) := h(x)/4\eta$.
With the block-encoding of $A$, we invoke \cref{lemma:polynomial_transformation_of_the_real_part_block_encoding} to construct a $(1, n+4, \delta_0)$-block-encoding of $g(A)$, called $\tilde U_g$, with $O(kn)$ circuit depth, using a total of $k-1$ calls to a controlled $U_{\Re{A}}$ and $U_{\Re{A}}^{\dagger}$ circuit, and with $O(poly(k, \log(1/\delta_0)))$ classical time complexity. 
There is a function $\tilde g(x)$ for which $\tilde U_g$ is a $(1, n+4, 0)$-block-encoding of $\tilde{g}(A)$  and $\lnorm{g(A) - \tilde{g}(A)}_2 \le \delta_0$. 
When we apply $\tilde U_g (\ket{0}_{n+4} \otimes U\ket{0}_n)$ and measure the $\ket{0}_{n+4}$ ancillary state, we obtain a state proportional to $\tilde{g}(A)\ket{\psi} = \sum_{j=1}^{2^n}\phi_j\ket{j}$, where we define the amplitudes of the unnormalized resulting state as $\{\phi_j\}_j$. 
We then define the normalization factor of this state as $\mathcal{N}_2 := \lnorm{\tilde{g}(A)\ket{\psi}}_2$.

Observing that $g(A)\ket{\psi} = \sum_j g(\psi_j)\psi_j \ket j=\frac{1}{4\eta}\sum_j P(\psi_j)\ket{j}$, we additionally define the unnormalized exact state as $\ket{P(\psi)} := g(A)\ket{\psi}$, with the corresponding normalization factor of $\mathcal{N}_1 := \lnorm{\ket{P(\psi)}}_2$. Thus, the exact normalized state we wish to prepare is $\frac{1}{\mathcal{N}_1}\ket{P(\psi)}$, and the normalized state we are able to prepare is $\frac{1}{\mathcal{N}_2}\tilde{g}(A)\ket{\psi}$. We will now bound the $\ell_2$-norm error of the produced state, and determine the probability of success of the procedure. 

First, we wish to determine the algorithmic complexity required to obtain an overall error $\epsilon$, defined as $\lnorm{\frac{1}{\mathcal{N}_1}\ket{P(\psi)} - \frac{1}{\mathcal{N}_2}\tilde{g}(A)\ket{\psi}}_2 \le \epsilon$. To do so, we must bound two quantities, $|\mathcal{N}_1 - \mathcal{N}_2|$ and $\lnorm{\ket{P(\psi)} - \tilde{g}(A)\ket{\psi}}_2$, so that \cref{lemma:normalized_state_deviation_due_to_normalization} may be applied. Using the definitions, we obtain
\begin{align} \label{eqn:midstep_transform_of_real_state_amplitudes_unnormalized_bound}
    \lnorm{\ket{P(\psi)} - \tilde{g}(A)\ket{\psi}}_2
    =
    \lnorm{g(A)\ket{\psi} - \tilde{g}(A)\ket{\psi}}_2 \le 
    \lnorm{g(A) - \tilde{g}(A)}_2 \leq \delta_0.
\end{align}
Now, we bound $|\mathcal{N}_1 - \mathcal{N}_2|$.
For any non-negative real numbers $a,b$, we have $|\sqrt{a} - \sqrt{b}| \le \sqrt{|a - b|}$. Consequently, $|\mathcal{N}_1 - \mathcal{N}_2| \le \sqrt{|\mathcal{N}_1^2 - \mathcal{N}_2^2|}$. 
By definition, $g(x)x = P(x)/4\eta =: f(x)$, and so
\begin{align}
|\mathcal{N}_1 - \mathcal{N}_2|
&\le 
\sqrt{
\left|
    \lnorm{\ket{P(\psi)}}_2^2 - \lnorm{\tilde{g}(A)\ket{\psi}}_2^2
\right|}
=
\sqrt{
\left|
    \sum_{j}(|f(\psi_j)|^2 - |\phi_j|^2)
\right|}
\\
&\le 
\sqrt{\max_j \left| |f(\psi_j)| - |\phi_j|\right|\sum_j(|f(\psi_j)| + |\phi_j|)}.
\end{align}
Since $||a| - |b|| \le |a - b|$, we have that $\max_j\left| |f(\psi_j)| - |\phi_j| \right| \le \max_j\left| f(\psi_j) - \phi_j \right|$. Moreover, $\sum_j |f(\psi_j) - \phi_j|^2=\lnorm{\ket{P(\psi)} - \tilde{g}(A)\ket{\psi}}_2^2 \le \delta_0^2$ (as per \cref{eqn:midstep_transform_of_real_state_amplitudes_unnormalized_bound})  which implies that $\max_j |f(\psi_j) - \phi_j| \le \delta_0$. 
Of course, $\lnorm{\tilde{g}(A)\ket{\psi}}_2^2 \le 1$ and so $\sum_j |\phi_j|^2 \le 1$ which implies that $\forall j, |\phi_j| \le 1$. 
Similarly, by definition, $f(x) = \frac{P(x)}{4\eta}$, and since $|P(x)| \le \eta$, $\forall j, f(\psi_j) \le 1/4$, and so $\sum_{j=1}^N |f(\psi_j)+\phi_j| \le 2N$. As a result,
\begin{align}\label{eqn:midstep_transform_of_real_state_amplitudes_norm_dif_bound}
|\mathcal{N}_1 - \mathcal{N}_2|
&\le 
\sqrt{2N\delta_0}.
\end{align}
Combining 
\cref{eqn:midstep_transform_of_real_state_amplitudes_unnormalized_bound}, \cref{eqn:midstep_transform_of_real_state_amplitudes_norm_dif_bound} and \cref{lemma:normalized_state_deviation_due_to_normalization}, we then obtain the following bound on the total $\ell_2$-norm error,
\begin{align}
    \lnorm{\frac{1}{\mathcal{N}_1}\ket{P(\psi)} - \frac{1}{\mathcal{N}_2}\tilde{g}(A)\ket{\psi}}_2
    \le 
    \frac{\sqrt{2N\delta_0} + \delta_0}{\mathcal{N}_1}
    \le 
    \frac{3\sqrt{N\delta_0}}{\mathcal{N}_1}.
\end{align}
We now analyze the probability of success of the procedure, thus determining the complexity in terms of the number of amplitude amplification steps required to succeed with arbitrarily high probability. Thus, we must lower-bound $\mathcal{N}_2^2 \equiv\lnorm{\tilde{g}(A)\ket{\psi}}_2^2$. 
We just showed that $|\mathcal{N}_1^2 - \mathcal{N}_2^2| \le 2N\delta_0$, hence $\mathcal{N}_2^2 \geq \mathcal{N}_1^2 - 2N\delta_0$.
Thus, requiring $\delta_0$ to be sufficiently small so as to ensure that $\mathcal{N}_1^2 > 2N\delta_0$, the complexity to boost to a constant probability of success is given by $O(\frac{1}{\sqrt{\mathcal{N}_1^2 - 2N\delta_0}})$. 

To obtain an overall error of at most $\epsilon$,
we set $\delta_0 = \epsilon^2\mathcal{N}_1^2/9N$. This indeed satisfies our requirement of $\mathcal{N}_1^2 > 2N\delta_0$ (for $\epsilon \le 1$).
As a result, the number of AA steps required can be simplified to  
$O\left(1/\mathcal{N}_1 \sqrt{1-\frac{2}{9}\epsilon^2}\right) \subseteq O\left(1/\mathcal{N}_1 \right)$, using $\epsilon \le 1$.
Since $\mathcal{N}_1=\frac{1}{4\eta}\mathcal{N}$, this bound becomes $O(\eta/\mathcal{N})$. Similarly, $\delta_0 = \epsilon^2\mathcal{N}_1^2/9N = \epsilon^2\mathcal{N}^2/144 \eta^2 N$.
As a result, with constant probability of success, our procedure produces a normalized quantum state $\epsilon$ close (in $\ell_2$-norm) to the state $\frac{1}{\mathcal{N}}\sum_j P(\psi_j)\ket{j}$ using a total of $O(kL/\mathcal{N})$ queries to a controlled $U$ circuit and controlled $U^{\dagger}$ circuit, with $O(knL /\mathcal{N})$ circuit depth, and with a total of $O(poly(k, \log(L^2 N/\epsilon^2\mathcal{N}^2)))$ classical computation.

Since $\mathcal{N}_1=\frac{1}{4\eta}\mathcal{N}$, and from \cref{lemma:upper_bound_on_max_of_reduced_polynomial} we have that $\eta \le L$, this bound becomes $O(L/\mathcal{N})$. Similarly, $\delta_0 = \epsilon^2\mathcal{N}_1^2/9N = \epsilon^2\mathcal{N}^2/144 \eta^2 N$, and since $\eta \le L$, we can set the more conservative bound $\delta_0 = \epsilon^2\mathcal{N}^2/144 L^2 N$.
As a result, with constant probability of success, our procedure produces a normalized quantum state $\epsilon$ close (in $\ell_2$-norm) to the state $\frac{1}{\mathcal{N}}\sum_j P(\psi_j)\ket{j}$ using a total of $O(kL/\mathcal{N})$ queries to a controlled $U$ circuit and controlled $U^{\dagger}$ circuit, with $O(knL /\mathcal{N})$ circuit depth, and with a total of $O(poly(k, \log(L^2 N/\epsilon^2\mathcal{N}^2)))$ classical computation. 
\end{proof}
The error bound assumes that the only error-source is the classical computation of the rotation angles in QSVT -- in practice there would likely be an additional logarithmic error from the implementation of all gates on real hardware, and so there would be an additional multiplicative factor in the complexity of approximately $\log(1/\mathcal{N}_1)$ in order to sufficiently suppress this error source and ensure numerical stability.

Note that in many cases where $k$ is small, it can actually be significantly advantageous to just implement the desired polynomial via products and linear combinations of block encodings of the $A$ matrix, as opposed to using QSVT, as the polynomial can be implemented \textit{exactly} in this manner (and for arbitrary complex initial states), albeit at the cost of using more ancillary qubits.

In \cref{appComparison}, we rederive the analogous result from \cite{guo2021nonlinear} with our diagonal block-encoding, and extend it with our $\ell_2$ norm guarantee.

\section{Non-Linear Transformations of State Amplitudes via Polynomial Approximations}\label{section:non_linear_transformations_via_polynomial_approximations}

In this section, we extend our previous result (\cref{theorem:polynomial_transformation_of_real_state_amplitudes_via_importance_sampling}) to arbitrary functions which can be approximated by polynomials. To this end, we begin by proving some simple results about functions and their polynomial approximations, and conclude with the main result of the section, \cref{theorem:generalized_end_to_end_complexity_for_f_0_is_0_for_efficiently_analytic_functions}.

\begin{lemma}\label{lemma:upper_bound_on_polynomial_norm_from_lipschitz_constant}
We are given a Lipschitz continuous real function $P(x)$ with Lipschitz constant $L$ on the interval $x\in [-1, 1]$, and the guarantee that $P(0)=0$. We are additionally given an arbitrary $\ell_2$ normalized quantum state $\ket{\psi} = \sum_{j=1}^N \psi_j \ket{\psi}$. Define the normalization factor $\mathcal{N}^2 := \sum_{j=1}^N |P(\psi_j)|^2$. Then, $\mathcal{N} \le L$.
\end{lemma}
\begin{proof}
First, by \cref{def:lipschitz_constant}, and $P(0) = 0$, for all $x\in [-1, 0)\cup (0, 1]$, $|P(x)| = |P(x) - P(0)| \le L |x| \le L$. 
Then, for any $j$, $|P(\psi_j)|^2 \le |x|^2 L^2$. Thus, since $\sum_{j=1}^N |\psi_j|^2 = 1$,
\begin{align}
    \mathcal{N}^2 = \sum_{j=1}^N |P(\psi_j)|^2
    \le \sum_{j=1}^N |\psi_j|^2 L^2 = L^2.
\end{align}
The result $\mathcal{N} \le L$ immediately follows, for the case where $x \neq 0$. Since the input state is $\ell_2$-normalized, there must be non-zero amplitudes, and these non-zero amplitudes must have a sum of squares equal to one. Thus, this result holds in full generality, and we have shown that $\mathcal{N} \le L$. 
\end{proof}

\begin{definition}
[Efficient uniform approximation]\label{def:efficiently_analytic_function}
Let  $f(x) : [-1,1]\to \mathbb R$ be a real-valued function in the domain $x\in[-1, 1]$.
With $k \in \mathbb Z_+$ and $\epsilon>0$, we say that the function $f(x)$ is $(\epsilon, k)$-uniformly approximated by a polynomial if there exists a degree-$k$ polynomial $P_k : [-1, 1] \to \mathbb{C}$  such that
\begin{align}
    \max_{x\in[-1, 1]}|f(x) - P_k(x)| \le \epsilon.
\end{align}
Equivalently, we say that $f(x)$ is an ``$\epsilon$-approximable function with degree $k$'', where $k$ can be any function of $\epsilon$.

\end{definition}
When it is clear from context, we will drop the subscript indicating the degree of the polynomial, i.e., sometimes we will write $P_k(x)$ as $P(x)$, and $e_k(x)$ as $e(x)$. We use the notation $e_k(x) := f(x) - P_k(x)$ for the error term. Note that the subscript on the error term does not correspond to the degree of the error term, rather it makes the dependence on the degree of approximating polynomial explicit. 
We show a couple simple technical lemmas regarding some properties of such functions and their corresponding normalizations.

\begin{lemma}[Composition of uniform approximations]\label{lemma:composition_of_polynomial_approximations}
    Let $f(x)$ be an $L_f$-Lipschitz, $(\epsilon_0, k_0)$-approximable function with approximation polynomial $P(x)$. Let $g(x)$ be an $(\epsilon_1, k_1)$-approximable function with approximation polynomial $Q(x)$. Then $f(g(x))$ is an $(\epsilon_0 + L_f\epsilon_1, k_0k_1)$-approximable function. I.e., $|f(g(x)) - P(Q(x))| \le \epsilon_0 + L_f \epsilon_1$.
\end{lemma}
\begin{proof}
    This follows simply from \cref{def:lipschitz_constant}, and \cref{def:efficiently_analytic_function}.
\end{proof}

\begin{lemma}\label{lemma:technical1}
Let $f(x)$ be $\epsilon_0$-approximable as per \cref{def:efficiently_analytic_function} with $P(x)$ the approximation polynomial and $\gamma := \max_{x\in [-1, 1]}|f(x)|$. Require that $\epsilon_0\leq \gamma$.
In addition, define the vector $\ket{\psi} := \sum_{j}\psi_j\ket{j}$, such that $\ket{\psi} \in \mathbb{C}^N$, and $\lnorm{\ket{\psi}}_2 = 1$. Define
$\mathcal{N}^2 := \sum_{j=1}^N|f(\psi_j)|^2$ and $\mathcal{N}_1^2 := \sum_{j=1}^N|P(\psi_j)|^2$. 
Then, 
\begin{itemize}
\item [i.)] $
    \max_{x\in[-1, 1]}|P(x)| \le 2\gamma$.
\item [ii.)] 
$\left| \mathcal{N} - \mathcal{N}_1
\right| \le 3\gamma\epsilon_0 N/\mathcal{N}$.
\end{itemize}
If $\epsilon_0 \le \frac{\mathcal{N}^2}{6\gamma N}$ (noting that $\frac{\mathcal{N}^2}{6\gamma N} < \gamma$), 
\begin{itemize}
\item [iii.)] 
$\mathcal{N}_1 \ge \frac{1}{2}\mathcal{N}$.
\end{itemize}
\end{lemma}

\begin{proof}
For proving i.), by \cref{def:efficiently_analytic_function}, $|P(x)| = |P(x) - f(x) + f(x)|\le \epsilon_0 + |f(x)|\le 2\gamma$, for all $x\in [-1,1]$.
For proving ii.), we obtain from simple algebra that $
\left|
    \mathcal{N} - \mathcal{N}_1
\right|
=
\left|
    \frac{\mathcal{N}^2 - \mathcal{N}_1^2}{\mathcal{N} + \mathcal{N}_1}
\right|
\le 
\frac{1}{\mathcal{N}}
\left|
\mathcal{N}^2 - \mathcal{N}_1^2
\right|$.
By definition, we have
$\left|
\mathcal{N}^2 - \mathcal{N}_1^2
\right|
\le 
\sum_{j=1}^N 
\left|
(f(\psi_j) - P(\psi_j))(f(\psi_j) + P(\psi_j))
\right|$.
By i.), we have that $|P(\psi_j)| \le 2\gamma$, and hence, $|f(\psi_j) + P(\psi_j)| \le 3\gamma$.  Since $f(x)$ is $\epsilon_0$-approximable, $|f(\psi_j) - P(\psi_j)|\le \epsilon_0$, and so we obtain
$
\left|
    \mathcal{N} - \mathcal{N}_1
\right|
\le 
3\gamma\epsilon_0 N/\mathcal{N}$.
For proving iii.) $\mathcal{N}_1^2 = \mathcal{N}^2 - \sum_{j}e(\psi_j)(P(\psi_j) + f(\psi_j)) \ge \mathcal{N}^2 - 3\epsilon_0 \gamma N$. If $\epsilon_0 \le \mathcal{N}^2/6\gamma N$, then $\mathcal{N}_1^2 \ge \mathcal{N}^2/2$.
\end{proof}

\begin{lemma}\label{lemma:l2_norm_error_of_polynomial_state_from_efficiently_analytic_function_state} 
Let $f(x)$ be $\epsilon_0$-approximable as per \cref{def:efficiently_analytic_function} with $P(x)$ the approximation polynomial and $\gamma := \max_{x\in [-1, 1]}|f(x)|$. 
Define the vector $\ket{\psi} = \sum_j \psi_j\ket{j}$ such that $\lnorm{\ket{\psi}}_2 = 1$ and $\ket{\psi} \in \mathbb{C}^N$. Let
$\mathcal{N}^2 := \sum_{j=1}^N|f(\psi_j)|^2$ and $\mathcal{N}_1^2 := \sum_{j=1}^N|P(\psi_j)|^2$.
Define $\epsilon \in (0,1]$ and 
$
\Delta_k := \lnorm{
\frac{1}{\mathcal{N}} 
\sum_j f(\psi_j)\ket{j}
-
\frac{1}{\mathcal{N}_1} 
\sum_j P_k(\psi_j)\ket{j}
}_2.
$
Then, $\epsilon_0 \le \epsilon\mathcal{N}^2/8\gamma N \implies \Delta_k \le \epsilon$.
\end{lemma}
\begin{proof}
Let $E := \mathcal{N}_1 - \mathcal{N}$.  Then, pulling out a common $\frac{1}{\mathcal{N}_1\mathcal{N}}$ term and simplifying by using $\mathcal{N}_1  = E + \mathcal{N}$ in the resulting numerator term, we get $\Delta \le (\epsilon_0\sqrt{N} + E)/\mathcal{N}_1$, where we also used $\mathcal{N} = \lnorm{\sum_j f(\psi_j)\ket{j}}_2$.

From \cref{lemma:technical1} ii.), we know that $|E| \le 3\gamma\epsilon_0 N/\mathcal{N}$. From \cref{lemma:technical1} iii.), we know that $\mathcal{N}_1 \ge \mathcal{N}/2$, if $\epsilon_0 \le \mathcal{N}^2/6\gamma N$. 
Noting that $\gamma\sqrt{N}/\mathcal{N} \ge 1$ implies that $\gamma \epsilon_0 N/\mathcal{N} = (\epsilon_0\sqrt{N})(\gamma\sqrt{N}/\mathcal{N})\ge \epsilon_0\sqrt{N}$, we then get,
\begin{align}
\Delta_k
&
\le 
\frac{2}{\mathcal{N}}
\left(
\epsilon_0\sqrt{N}
+
\frac{3\gamma\epsilon_0 N}{\mathcal N}
\right) \leq \frac{8\epsilon_0 \gamma N}{\mathcal N^2}.
\end{align}
Thus, to get an overall error at most $\epsilon > 0$, we require that $\frac{8\epsilon_0 \gamma N}{\mathcal N^2} \le \epsilon$ and so we set $\epsilon_0 \le \epsilon\mathcal{N}^2/8\gamma N$. 
Since our proof added the constraint $\epsilon_0 \le \mathcal{N}^2/6\gamma N$ (and technically also $\epsilon_0 \le \gamma$, but since $\mathcal{N}^2/6\gamma N \le \gamma$, this automatically satisfied), $\epsilon_0 = \mathcal{N}^2\epsilon/8\gamma N \le \mathcal{N}^2/6\gamma N$ and so our proof holds without further constraint if $\epsilon \le 4/3$.
\end{proof}

\begin{theorem}\label{theorem:generalized_end_to_end_complexity_for_f_0_is_0_for_efficiently_analytic_functions}
We are given an $n$-qubit circuit $U$ such that $U\ket{0} = \sum_{j=1}^N\psi_j\ket{j}$, with $N = 2^n$ and $\forall j, \psi_j\in \mathbb R$. Additionally, we are given a function $f(x)$ with $f(0) = 0$  and $\gamma:=\max_{x\in[-1,1]}|f(x)|$ and $\mathcal{N}^2 := \sum_{j=1}^N |f(\psi_j)|^2$. 
For $1 > \epsilon >0$, let $f(x)$ be 
$\frac{\gamma N}{\epsilon\mathcal{N}^2}$-approximable by a polynomial $P_k(x)$, such that $P_k(0) = 0$, with degree 
$k = K(\epsilon\mathcal{N}^2/ \gamma N)$, where $K$ is a function describing the degree of the polynomial in terms of the error.
Additionally, define $\tilde \gamma := \max_{x \in [-1, 1]}|P_k(x)/x|$.  
Then, we can prepare a state $\epsilon$ close in $\ell_2$ norm-distance from 
$\frac{1}{\mathcal{N}}\sum_j f(\psi_j)\ket{j}$ with arbitrarily high probability of success with $O(\tilde \gamma k/\mathcal{N})$ query complexity to a controlled $U$ and $U^{\dagger}$ circuit, with overall circuit depth of $O(n \tilde \gamma k /\mathcal{N})$, with $O({\rm poly} \log(\tilde \gamma^2 N/ \epsilon^2 \mathcal{N}^2))$ classical computation and with $O(n)$ ancillary qubits.
\end{theorem}
\begin{proof}
We will bound the $\ell_2$-norm error from the polynomial approximation, and then the $\ell_2$-norm error from the algorithmic error, combining the two via triangle inequality. Define $\mathcal{N}_1^2 := \sum_{j=1}^N |P_k(\psi_j)|^2$.
By \cref{lemma:l2_norm_error_of_polynomial_state_from_efficiently_analytic_function_state}, if $\epsilon/2 \le \min(1, 1/\gamma)$, then 
\begin{align}
    \lnorm{
    \frac{1}{\mathcal{N}} 
    \sum_j f(\psi_j)\ket{j}
    -
    \frac{1}{\mathcal{N}_1} 
    \sum_j P_k(\psi_j)\ket{j}
    }_2
    \le 
    \epsilon/2.
\end{align}
Where $\max_{x\in[-1, 1]}|f(x) - P_k(x)| \le \epsilon_0$, \cref{lemma:l2_norm_error_of_polynomial_state_from_efficiently_analytic_function_state} says $\epsilon/2 = 8\epsilon_0\gamma N/\mathcal{N}^2$, consequently $\epsilon_0 = \epsilon\mathcal{N}^2/16\gamma N \le \mathcal{N}^2/6\gamma N$, and so we can apply \cref{lemma:technical1}.
This gives $\mathcal{N}_1 \ge \mathcal{N}/2$.

We will invoke \cref{theorem:polynomial_transformation_of_real_state_amplitudes_via_importance_sampling} on $P_k(x)$, thereby obtaining the approximation to the desired transformation. Thus, we can prepare the state $\frac{1}{\mathcal{N}_1}\sum_j P_k(\psi_j)\ket{j}$ up to error $\epsilon/2$ by invoking \cref{theorem:polynomial_transformation_of_real_state_amplitudes_via_importance_sampling} with $P_k(x)$. This has a query complexity to a controlled $U$ and controlled $U^{\dagger}$ circuit of $O(\tilde \gamma h(\gamma N/\epsilon\mathcal{N}^2)/\mathcal{N})$. The additional circuit complexity is thus $O(n \tilde \gamma h(\gamma N/\epsilon\mathcal{N}^2)/\mathcal{N})$, requiring $O({\rm poly}(\log(\gamma N/\epsilon\mathcal{N}^2), \log(\tilde \gamma^2 N/ \epsilon^2 \mathcal{N}^2)))$ classical computation, and $O(n)$ ancilla qubits. To increase the simplicity of the reported result, we assume that $\epsilon$ is sufficiently small such that $\log(\gamma N/\epsilon\mathcal{N}^2) \le \log(\tilde \gamma^2 N/ \epsilon^2 \mathcal{N}^2)$, allowing the classical computation complexity to be simplified to $O({\rm poly}\log(\tilde \gamma^2 N/ \epsilon^2 \mathcal{N}^2))$.

\end{proof}
When $K(\epsilon) = \log(1/\epsilon)$, the query and circuit complexities simplify accordingly.

\section{Application: End-to-End Complexities for Applying Various Functions to Arbitrary Quantum States}\label{section:end_to_end_complexity_for_applying_various_functions}

We now use the theorems and lemmas derived so far to easily derive the end-to-end complexity for a number of common functions. When $f(0)=0$, we invoke \cref{theorem:generalized_end_to_end_complexity_for_f_0_is_0_for_efficiently_analytic_functions}. When $f(0) > 0$, we invoke \cref{theorem:generalized_end_to_end_complexity_for_previous_technique_for_efficiently_analytic_functions} (our diagonal variant of \cite{guo2021nonlinear}, extended with $\ell_2$-norm error guarantees, as derived in the appendix).
These examples illustrate how when combined with our new error bound, our technique and the technique of \cite{guo2021nonlinear} enable the transformation of real amplitudes by a rich set of functions \textit{with overall efficient end-to-end complexity}.

\begin{lemma}[Absolute error in polynomial approximation to tanh \cite{guo2021nonlinear}]\label{lemma:polynomial_approximation_to_tanh_absolute_error_bound}
Let $f(x) = \tanh(x)$. Define
\begin{align}
    P_k(x) := \sum_{n=1}^{k}\frac{2^{2n}(2^{2n} - 1)B_{2n}}{(2n)!}x^{2n -1},
\end{align}
where $B_{n}$ is the $n^{th}$ Bernoulli number. Then, in the interval $x \in [-1, 1]$, when $k\ge 2$,
\begin{align}
     \left|f(x) - P_k(x)\right| \le 9 \left(\frac{2}{\pi}\right)^{k}.
\end{align}
Thus, there exists a $k\in O(\log(1/\epsilon))$ such that the error is at most $\epsilon$.
\end{lemma}
\begin{proof}
The proof follows directly from \cite{guo2021nonlinear}, with (very) minor corrections. 
First, define $\alpha_n := \frac{2^{2n}(2^{2n} - 1)B_{2n}}{(2n)!}$. Note that in the interval $x\in [-\pi/2, \pi/2]$, $\tanh(x) = \sum_{n=1}^{\infty}\alpha_n x^{2n - 1}$. Define $e_k(x) := f(x) - P_k(x)$. Assuming $x\in[-1,1]$, wish to bound,
\begin{align}
    |e_k(x)|
    =
    \left|
        \sum_{n=k+1}^{\infty}\alpha_n x^{2n - 1}
    \right|
    \le 
    \sum_{n=k+1}^{\infty}|\alpha_n|.
\end{align}
Using the fact that $|B_{2n}| \le 5\sqrt{\pi n}(\frac{n}{\pi e})^{2n}$ as per~\cite{leeming1987real} (as noted in ~\cite{guo2021nonlinear}), 
\begin{align}
    \left|
        \alpha_n
    \right|
    \le
    \left|
        \frac{2^{4n}B_{2n}}{(2n)!}
    \right|
    \le 
    \frac{5\sqrt{\pi n}}{(2n)!}\left(\frac{4n}{\pi e}\right)^{2n}.
\end{align}
If we wish to use the polynomial approximation $P_1(x)$, we are just implementing an identity transformation, and so the minimum polynomial we might use is $P_2(x)$. As a result, without a loss of generality, we can assume that the $n$ in the error term $\sum_n^{\infty}\alpha x^{2n - 1}$ is at least $3$. Noting that $\sqrt{\pi n} \le (\pi/2)^n$ for all $n \ge 3$, and that $n! \ge (n/e)^n$, 
\begin{align}
    \left|
        \alpha_n
    \right|
    \le
    5\sqrt{\pi n}\left(\frac{2}{\pi}\right)^{2n}
    \le
    5\left(\frac{2}{\pi}\right)^{n}.
\end{align}
This then gives the bound
\begin{align}
    |e_k(x)|
    \le 
    5\sum_{n=k+1}^{\infty}\left(\frac{2}{\pi}\right)^{n}
    \le 9 \left(\frac{2}{\pi}\right)^{k}.
\end{align}
\end{proof}

\begin{theorem}[End-to-end complexity for several functions]\label{theorem:end_to_end_complexity_for_applying_various_functions}
We are given an arbitrary $n$-qubit quantum state preparation unitary $U$ such that $U\ket{0} = \ket{\psi} = \sum_{j=1}^N\psi_j\ket{j}$ with $\forall j, \psi_j \in \mathbb{R}$
and $N = 2^n$.
Let $\epsilon>0$, for the functions $f(x)$ given below there exists circuits to prepare a state $\ket \phi$ such that $\Vert \ket \phi - \frac{1}{\mathcal{N}}\sum_j f(\psi_j)\ket j\Vert_2 \leq \epsilon$,  where $\mathcal{N}$ is the normalization factor. The circuits all require $O(n)$ ancillas, and their complexities are:\\
\begin{center}
\begin{tabular}{c|c|c|c}
\text{Case} & f(x) & \text{Query complexity} & \text{Depth} \\
\hline
i.) & $e^x$ & $O\left(\log\left(\frac{1}{\epsilon}\right)\right)$ & $O\left(n\log\left(\frac{1}{\epsilon}\right)\right)$ \\
ii.)& $\cos(x)$ & $O\left(\log\left(\frac{1}{\epsilon}\right)\right)$ & $O\left(n\log\left(\frac{1}{\epsilon}\right)\right)$ \\
iii.) & $\frac{1}{1 + e^{-2 x}}$ & $O\left(\log\left(\frac{1}{\epsilon}\right)\right)$ & $O\left(n\log\left(\frac{1}{\epsilon}\right)\right)$ \\
iv.) & $\frac{1}{\sigma\sqrt{2\pi}}e^{-\frac{1}{2}(x/\sigma)^2}$ ($\sigma^2 \ge 1/2$) & $O\left (\sigma\log\left(\frac{\sigma^2}{\epsilon}\right)\right)$ &  $O\left(n\sigma\log\left(\frac{\sigma^2}{\epsilon}\right)\right)$ \\
v.) & $\sin(x)$ & $O\left(\log\left(\frac{N}{\epsilon}\right)\right)$ & $O\left(n\log\left(\frac{N}{\epsilon}\right)\right)$ 
\end{tabular}
\end{center}
The query complexity is to controlled $U$ and $U^{\dagger}$ circuits. We assume that $x\in [-1, 1]$.
\end{theorem}
\begin{proof}
i.) Using the standard series representation of $f(x)= e^x$, for $x\in[-1, 1]$,
$
f(x) = \sum_{n=0}^{\infty}\frac{x^n}{n!}.
$
We have for the error term that
$\left|
f(x) - P_k(x)
\right|
=
\left|
\sum_{n=k+1}^{\infty}\frac{x^n}{n!}
\right|
\le 
\sum_{n=k+1}^{\infty}\frac{1}{n!}
\le 
\sum_{n=k+1}^{\infty}2^{-n}
=
2^{-k},
$
using the fact that for all $n \ge 3$, $1/n! \le (1/2)^n$. Moreover, $\forall x\in[-1, 1]$, $1/e \le e^x \le e$. As a result, $\mathcal{N}^2 = \sum_j f(\psi_j)^2 \ge N/e^2$, and so $\mathcal{N} \ge \sqrt{N}/e$. Consequently, $O(\log(\gamma N/\epsilon \mathcal{N}^2)\gamma\sqrt{N}/\mathcal{N})$ is  $O(\log(1/\epsilon))$, giving a query complexity of $O(\log(1/\epsilon))$ and an associated circuit depth of $O(n\log(1/\epsilon))$ through~\cref{theorem:generalized_end_to_end_complexity_for_previous_technique_for_efficiently_analytic_functions}.   
\\
ii.) Using the standard series representation of $f(x)=\cos(x)$, for $x\in[-1, 1]$,
$
f(x) = \sum_{n=0}^{\infty}(-1)^n \frac{x^{2n}}{(2n)!}.
$
The error term is
$
|f(x) - P_k(x)|
=
\left|
\sum_{n=k+1}^{\infty}(-1)^n \frac{x^{2n}}{(2n)!}
\right|
\le 
\sum_{n=k+1}^{\infty}
\frac{1}{(2n)!}
\le 
\sum_{n=k+1}^{\infty}
\left(\frac{1}{2}\right)^n
= 
2^{-k}$,
using the fact that for $n \ge 2$, $(2n)! \ge 2^n$. Thus, $|e_k(x)| \le \epsilon_0$ with overall $k \in O(\log(1/\epsilon_0))$.
Of course, $\min_{x\in[-1, 1]}|\cos(x)| = \cos(1) > 1/2$, and so $\mathcal{N}^2 = \sum_j \cos(\psi_j)^2 \ge N/4$ which implies that $\mathcal{N} \ge \sqrt{N}/2$. 
Consequently, $O(\log(\gamma N/\epsilon \mathcal{N}^2)\gamma\sqrt{N}/\mathcal{N})$ is  $O(\log(1/\epsilon))$, giving a query complexity of $O(\log(1/\epsilon))$ and an associated circuit depth of $O(n\log(1/\epsilon))$ through~\cref{theorem:generalized_end_to_end_complexity_for_previous_technique_for_efficiently_analytic_functions}. 
iii.) We use the common identity for the logistic function of
$
f(x) = \frac{1}{2} + \frac{1}{2}\tanh(x),
$
Consequently, to obtain a series expansion for $f(x)$, we invoke \cref{lemma:polynomial_approximation_to_tanh_absolute_error_bound}, again adopting the notation $\alpha_n := \frac{2^{2n}(2^{2n} - 1)B_{2n}}{(2n)!}$, yielding
$f(x) = 
\frac{1}{2}\left(
1 + \sum_{n=1}^{\infty}\alpha_{n}x^{2n - 1}
\right)$.
Consequently, letting $P_k(x) := (1/2)(1 + \sum_{n=1}^{k}\alpha_{n}x^{2n - 1})$, and defining $e_k(x) := f(x) - P_k(x)$, we obtain 
$
|e_k(x)| = \left|
\frac{1}{2}\sum_{n=k+1}^{\infty} \alpha_n x^{2n - 1}
\right|
\le 
5\left(\frac{2}{\pi}\right)^k,
$
where the bound follows directly from \cref{lemma:polynomial_approximation_to_tanh_absolute_error_bound}, and again we require that $k\ge 2$. Moreover, $|f(x)| \ge 1/(1 + e^2) > 1/10$, thus $\mathcal{N} = \sqrt{\sum_{j}f(\psi_j)^2} \ge \sqrt{N}/10$.
Consequently, $O(\log(\gamma N/\epsilon \mathcal{N}^2)\gamma\sqrt{N}/\mathcal{N})$ is  $O(\log(1/\epsilon))$, giving a query complexity of $O(\log(1/\epsilon))$ and an associated circuit depth of $O(n\log(1/\epsilon))$ through~\cref{theorem:generalized_end_to_end_complexity_for_previous_technique_for_efficiently_analytic_functions}. 
iv.) Using the fact that $e^x= \sum_{n=0}^{\infty}\frac{x^n}{n!}$,  defining $\alpha := -1/2\sigma^2$ and $\beta := \frac{1}{\sigma\sqrt{2\pi}}$, we readily obtain
$
    f(x) = \beta e^{\alpha x^2} = \beta \sum_{n=0}^{\infty}\frac{\alpha^n x^{2n}}{n!}.
$A
Defining the first $k + 1$ terms of this series as $P_k(x)$, and defining $e_k(x) := f(x) - P_k(x)$, when $x\in[-1, 1]$, for all $k\ge 3$,
\begin{align}
    |e_k(x)|
    =
    \beta \left|
        \sum_{n=k+1}^{\infty}\frac{\alpha^n x^{2n}}{n!}
    \right|
    \le 
    \frac{1}{\pi}
    \sum_{n=k+1}^{\infty}\frac{|\alpha|^n}{n!}
    \le 
    \frac{1}{\pi}
    \sum_{n=k+1}^{\infty}\frac{1}{n!}
    \le 
    \frac{1}{\pi}
    \sum_{n=k+1}^{\infty}(1/2)^n
    =
    \frac{2^{-k}}{\pi},
\end{align}
following from assumption, since $\sigma^2 \ge 1/2$, and so $|\alpha| \le 1$. Moreover, since $\forall x \in[-1, 1], |f(x)| \ge 1/\sigma\sqrt{2\pi}$ (since $\sigma^2 \ge 1/2$), $\mathcal{N}=\sqrt{\sum_jf(\psi_j)^2}\ge \sqrt{N}/\sigma\sqrt{2\pi}$.
Consequently, $O(\log(\gamma N/\epsilon \mathcal{N}^2)\gamma\sqrt{N}/\mathcal{N})$ is  $O(\sigma\log(\sigma^2/\epsilon))$, giving a query complexity of $O(\sigma\log(\sigma^2/\epsilon))$ and an associated circuit depth of $O(n\sigma\log(\sigma^2/\epsilon))$ through~\cref{theorem:generalized_end_to_end_complexity_for_previous_technique_for_efficiently_analytic_functions}. 
v.) First, we use the series representation for $f(x) = \sin(x)$ in the interval $[-1, 1]$ of 
$
f(x) = \sum_{n=0}^{\infty}(-1)^n\frac{x^{2n+1}}{(2n+1)!}. 
$
The error term is
$
|e_k(x)| = |f(x) - P_k(x)| = \left|\sum_{n=k+1}^{\infty}(-1)^n\frac{x^{2n+1}}{(2n+1)!}\right|
\le 
\sum_{n=k+1}^{\infty}\left|\frac{1}{(2n+1)!}\right|
\le 
\sum_{n=k+1}^{\infty}\frac{1}{2^n} = 2^{-k},
$
using the fact that for $n \ge 4$, $n! \ge 2^n$.  We now prove that $\tilde \gamma := \max_{x\in[-1, 1]}|P'_k(x)| \le 2$.
Noting that $P'_k(x) = \sum_{n=0}^k \frac{(-1)^n x^{2n}}{(2n)!}$, we can observe that this is the degree $k$ truncation of the polynomial representation of $\cos(x)$, allowing us to upper-bound this value with 
\cref{lemma:technical1}, giving $\tilde \gamma \le 2$. 
Obviously, $\max_{x\in[-1, 1]}|f(x)| \le 1$. Additionally, $\lim_{x\to 0}|\sin(x)|/|x| = 1$, and for all $x\in[-1, 1]$, $|\sin(x)|/|x| \ge 3/4$, implying that $\mathcal{N}^2 = \sum_{j} |\sin(\psi_j)|^2 \ge \frac{9}{16}$. 
As a result, $O(\tilde \gamma \log(\gamma N/\epsilon \mathcal{N}^2)/\mathcal{N})$ is  $O(\log(N/\epsilon))$, giving a query complexity of $O(\log(N/\epsilon))$ and an associated circuit depth of $O(n\log(N/\epsilon))$ through~\cref{theorem:generalized_end_to_end_complexity_for_f_0_is_0_for_efficiently_analytic_functions}. 
\end{proof}

\section{Application: Exponential Improvement in Applying Tanh Activation Function in Machine Learning}\label{section:exponential_improvement_tanh}
To clearly demonstrate the exponential improvement offered by our technique over that of~\cite{guo2021nonlinear}, we give the end-to-end complexity statement for both algorithms applied to the important function $f(x) = \tanh(x)$ (using \cref{theorem:modified_variant_of_previous_technique_complexity_of_applying_a_polynomial_to_a_state} for consistency).

\begin{theorem}
[Exponential speedup in importance-weighted application of $\tanh$ activation function]
\label{theorem:exponential_speedup_in_applying_tanh_demonstration}
Given an arbitrary $n$-qubit quantum state preparation unitary $U$ such that $U\ket{0} = \ket{\psi} = \sum_{j=1}^N\psi_j\ket{j}$ (with $\forall j, \psi_j \in \mathbb{R}$, and $N = 2^n$), and the function $f(x)= \tanh(x)$, the algorithm of \cite{guo2021nonlinear}, as described in \cref{theorem:modified_variant_of_previous_technique_complexity_of_applying_a_polynomial_to_a_state}, prepares a state $\epsilon$ close in $\ell_2$-norm distance (for small $\epsilon$) to $\frac{1}{\mathcal{N}}\sum_jf(\psi_j)\ket{j}$ (where $\mathcal{N}$ is the normalization factor) with arbitrarily high probability of success with worst-case $O(\sqrt{N}\log(N/\epsilon))$ query complexity to a controlled $U$ and $U^{\dagger}$ circuit, with overall circuit depth of $O(n\sqrt{N}\log(N/\epsilon))$, and $O(n)$ ancillary qubits.
Our importance-sampling based procedure prepares a state $\epsilon$ close in $\ell_2$-norm distance from $\frac{1}{\mathcal{N}}\sum_jf(\psi_j)\ket{j}$ with arbitrarily high probability of success with worst-case $O(\log(N/\epsilon))$  query complexity, with overall circuit depth of $O(n\log(N/\epsilon))$, and $O(n)$ ancillary qubits.
\end{theorem}
\begin{proof}

\Cref{lemma:polynomial_approximation_to_tanh_absolute_error_bound} implies that $f(x)$ is an efficiently analytic function as per \cref{def:efficiently_analytic_function}. Defining $\alpha_n$ as per \cref{lemma:polynomial_approximation_to_tanh_absolute_error_bound}, we can write $f(x) = \sum_{n=1}^{\infty}\alpha_n x^{2n - 1}$, $P_k(x) = \sum_{n=1}^{k}\alpha_n x^{2n - 1}$ and $e_k(x) = \sum_{n=k+1}^{\infty}\alpha_n x^{2n - 1}$.

We begin by proving three facts about $f(x)$: (1) $\mathcal{N} \ge 3/4$, (2) $|f(x)| \le 1$, and (3) $\tilde \gamma \le 2$ (for $k \ge 15$). First, noting that $|f(x)|/|x|\ge 3/4$ (with the minimum at $x = \pm 1$) implying that $|f(x)| \ge 3|x|/4$, we have $\mathcal{N}^2 = \sum_{j}|f(\psi_j)|^2 \ge (9/16)\sum_{j}|\psi_j|^2 = 9/16$. Thus, $\mathcal{N} \ge 3/4$. Second, trivially $|f(x)| \le 1$. Finally, we prove that $\tilde \gamma \le 2$.
By \cref{lemma:upper_bound_on_max_of_reduced_polynomial}, $\tilde \gamma := \max_{x\in[-1, 1]}|P_k(x)/x| \le \max_{x\in[-1, 1]}|\frac{d}{dx}P_k(x)|$. To bound $\tilde \gamma$, we will first compute the Lipschitz constant of $f(x) = \tanh(x)$ in the interval $x\in[-1, 1]$, $L$:
\begin{align}\label{eqn:lipschitz_constant_tanh}
    L = \max_{x\in[-1, 1]}\left|
    \frac{d}{dx}
        \tanh(x)
    \right|
    = 
    \max_{x\in[-1, 1]}
    \left |
        1 - \tanh^2(x)
    \right|
    =
    1.
\end{align}
Then,
\begin{align}
    \left|\frac{d}{dx}P_k(x)\right|
    \le 
    \left|\frac{d}{dx}f(x)\right|
    +
    \left|\frac{d}{dx}e_k(x)\right|
    \le 
    1 + 
    \sum_{n = k + 1}^{\infty} (2n-1)\left(\frac{2}{\pi} \right)^{n}.
\end{align}
Observing that $\left(\frac{\pi}{2}\right)^{n/2} \ge (2n - 1)$ for all $n\ge 16$, 
\begin{align}
    \tilde \gamma
    \le 
    1 + 
    \sum_{n = k + 1}^{\infty} \left(\frac{2}{\pi} \right)^{n/2}
    \le 
    1 + 2 \left(\frac{2}{\pi}\right)^k
    \le 
    2,
\end{align}
imposing the restriction that $k \ge 15$.

We now branch the proof into two cases, one using the importance sampling technique, and the other using the $\ell_2$-norm version of \cite{guo2021nonlinear}. In both algorithms, we aim to produce an $\ell_2$-normalized quantum state $\ket{\phi}$ such that $\lnorm{\frac{1}{\mathcal{N}}\sum_jf(\psi_j)\ket{j} - \ket{\phi}}_2 \le \epsilon$ for some $\epsilon \ge 0$.
Noting that for a fixed function the classical computation can be precomputed and so we neglect the classical computation cost in our complexity statement.
To get our result with the importance sampling based technique, we invoke \cref{theorem:generalized_end_to_end_complexity_for_f_0_is_0_for_efficiently_analytic_functions} finding a query complexity (to a controlled $U$ and $U^{\dagger}$ circuit) of $O(\log(N/\epsilon))$, and a circuit complexity of $O(n\log(N/\epsilon))$. We get the result with the previous technique by invoking~\cref{theorem:generalized_end_to_end_complexity_for_previous_technique_for_efficiently_analytic_functions}, finding a query complexity of $O(\sqrt{N}\log(N/\epsilon))$ and a circuit complexity of $O(n\sqrt{N}\log(N/\epsilon))$. Thus, we have achieved an exponential improvement in the end-to-end complexity of applying the $\tanh(x)$ function to an arbitrary quantum state over the previous technique.
\end{proof}

\section{Application: Maximum Finding}\label{section:maximum_finding}

In this section, we derive a technique for maximum finding in the state preparation unitary input model. 
We expect that variants of this technique will be of significant interest to domains ranging from financial derivative pricing~\cite{chakrabarti2021threshold} to real-space molecular chemistry simulations~\cite{chan2023grid}.
This algorithm is made possible by our results on non-linear amplitude transformations, and in particular our importance-weighted results.
Given a state preparation unitary $U$, (as per \cref{def:state_prep_unitary_input}) preparing the $\ell_2$-normalized quantum state $U\ket{0} = \sum_{j=0}^{N-1}\psi_j\ket{j}$, the objective is to identify the basis vector associated with the maximum amplitude, with constant probability of success. 
This problem format has some overlap with the state tomography results of~\cite{van2021quantum, van2023quantum}, but the examination of this overlap remains as future work.

We require two results from~\cite{low2017hamiltonian}: the error function approximation to the sign function, and the polynomial approximation to the shifted error function.

\begin{lemma}[Shifted and scaled error function approximation to the shifted Heavyside function~\cite{low2017hamiltonian}]\label{lemma:shifted_erf_approximation_to_shifted_heavyside}
Let $\delta > 0$, $m>0$, $\tau \in [-1, 1]$, and $\epsilon \in (0, 1]$. Define the shifted heavy function in terms of the sign function, i.e., $\hvy_{\tau}(x) := \frac{1}{2}(\sgn(x - \tau) + 1)$.
Define the shifted and scaled error function as $\erfs_{m, \tau}(x) = \frac{1}{2}(\erf(m(x - \tau)) + 1)$. When $m = \frac{\sqrt{2}}{\delta}\log^{1/2}(\frac{2}{\pi\epsilon^2})$,
\begin{align}
    \max_{x\in [-1, \tau - \delta/2]\cup [\tau + \delta/2, 1]}|\erfs_{m, \tau}(x) - \hvy_{\tau}(x)| \le \epsilon,
\end{align}
\end{lemma}
\begin{proof}
    This is another way of stating Lemma 10 of~\cite{low2017hamiltonian}.
\end{proof}

\begin{lemma}[Polynomial approximation to the shifted and scaled error function~\cite{low2017hamiltonian}]\label{lemma:polynomial_approximation_to_shifted_error_function}
Let $\tau \in [-1, 1]$, $m > 0$ and $\epsilon \in (0, O(1)]$. Then, we can efficiently construct a degree-$k$ polynomial, $P_{erf, m, \tau, k}(x)$, satisfying
\begin{align}
    \max_{x\in[-1, 1]}|P_{erfs, m, \tau, k}(x) - \erfs_{m, \tau}(x)| \le \epsilon
\end{align}
with $k \in O(m\log(1/\epsilon))$. 
Moreover, let $y = m(x - \tau)$, then
\begin{align}
    P_{erfs, m, \tau, k}(x) := 
    \frac{1}{2}
    +
    \frac{2m e^{-2m^2}}{\sqrt{\pi}}
    \left(
        I_0(2m^2)y
        +
        \sum_{j=1}^{(n-1)/2}
        I_j(2m^2)(-1)^j
        \left(
            \frac{T_{2j+1}(y)}{2j+1}
            -
            \frac{T_{2j-1}(y)}{2j-1}
        \right)
    \right)
\end{align}
where $I_j(\cdot)$ is a modified Bessel function of the first kind, and $T_n(\cdot)$ is a Chebyshev polynomial of the first kind.
\end{lemma}
\begin{proof}
Follows immediately from Corollary 5 of ~\cite{low2017hamiltonian}. 
\end{proof}

\begin{theorem}[Maximum finding in state preparation  input model]\label{theorem:maximum_finding_in_state_preparation_unitary_input_model}
Let $U$ be an $n$-qubit state preparation unitary $U\ket{0} = \sum_{j=1}^{N}\psi_j \ket{x_j}$ such that $\forall j: \psi_j \in \mathbb{R}_{\ge 0}$, and $\{x_j\}_j$ is some ordering on the set of standard basis vectors, where the subscript makes the mapping to the corresponding amplitude label explicit. Without loss of generality, assume that $\psi_1 > \psi_2 \ge ... \ge \psi_N$, and define the amplitude gap $\Delta := \psi_1 - \psi_2$. 
We are given the values $\psi_1$ and $\Delta$. 
Then, with arbitrarily high probability, and $\tilde{O}(\frac{1}{\psi_{1}\Delta}$) query complexity to a controlled $U$ and $U^{\dagger}$ circuit, we can determine the index of the basis state with maximum probability amplitude (i.e., we learn which basis state the label $x_1$ corresponds to). 
\end{theorem}
\begin{proof}

Let $m > 0$, $\epsilon>0$, $\epsilon_0 > 0$ and $\epsilon_1 > 0$. Set the parameters $\delta$ and $\tau$ in \cref{lemma:shifted_erf_approximation_to_shifted_heavyside} as $\delta = \Delta$, and $\tau := \psi_1 - \frac{\Delta}{2}$. \Cref{lemma:polynomial_approximation_to_shifted_error_function} proves that the function $\erfs_{m, \tau}(x)$ is $\epsilon_0$-uniformly approximated by the polynomial $P_{m, k}(x) := P_{erfs, m, \tau, k}(x)$ of degree $k \in O(m\log(1/\epsilon_0))$. Define the function $f_{m, \tau}(x) := x\erfs_{m, \tau}(x)$. The function $f_{m, \tau}(x)$ is $\epsilon_0$-uniformly approximated by the polynomial $x P_{m, k}(x)$, and has $\gamma := \max_{x\in[-1, 1]}|f(x)| \le 1$. Finally, $\tilde \gamma = \max_{x\in[-1, 1]}|x P_{m, k}(x) / x| = \max_{x\in[-1, 1]}|P_{m, k}(x)| \le 2$ (for $\epsilon_0 \le 1$).
Define $\mathcal{N}^2 := \sum_{j}f_{m, \tau}(\psi_j)^2$.
Setting $m = \frac{\sqrt{2}}{\Delta}\log^{1/2}(\frac{2}{\pi \epsilon_1^2})$, \cref{lemma:shifted_erf_approximation_to_shifted_heavyside} ensures that $\erfs_{m, \tau}(\psi_1) \ge 1 - \epsilon_1$, and that $\forall j\neq 1: \erfs_{m, \tau}(\psi_j) \le \epsilon_1$. Consequently, $\mathcal{N}^2 \ge f_{m, \psi_1}(\psi_1)^2 \ge (1-\epsilon_1)^2\psi_1^2 \implies \mathcal{N} \ge (1-\epsilon_1)\psi_1$. Assuming that $\epsilon_1 \le 1/2$, $\mathcal{N} \ge \psi_1/2$. 
Finally, define the $\ell_2$-normalized state $\ket{f_{m, \tau}(\psi)} := \frac{1}{\mathcal{N}}\sum_j f_{m, \tau}(\psi_j)\ket{x_j}$. 

We want to determine the query and circuit complexities required to obtain the $\ell_2$-normalized state $\ket{\phi} = \sum_{j}\phi_j\ket{x_j}$ such that $\lnorm{\ket{\phi} - \ket{x_1}}_2 \le \epsilon$, so that we can then use it to bound the complexity of identifying the basis vector with maximum amplitude, with high probability of success.

First, we will prove that by setting $m = \frac{\sqrt{2}}{\Delta}\log^{1/2}(\frac{16\sqrt{N}}{\pi \epsilon^2})$, we get the guarantee that $\lnorm{\ket{f_{m, \tau}(\psi)} - \ket{x_1}}_2 \le \epsilon/2$. Using $\mathcal{N}^2 \le \psi_1^2(\erfs_{m, \tau}(\psi_1)^2 + \epsilon_1^2 N)$, and that for non-negative numbers $a, b$, $\sqrt{a^2 + b^2} \le a + b$, we have $\mathcal{N} \le \psi_1 (\erfs_{m, \tau}(\psi_1) + \epsilon_1\sqrt{N})$, and so
\begin{align}
    \lnorm{\ket{f_{m, \tau}(\psi)} - \ket{x_1}}_2^2
    &=
    2\left(
    1 - \frac{\psi_1\erfs_{m, \tau}(\psi_1)}{\mathcal{N}}
    \right)
    \le 
    2\left(
    \frac{\sqrt{N}\epsilon_1}{\erfs_{m,\tau}(\psi_1) + \sqrt{N}\epsilon_1}
    \right).
\end{align}
Since $\erfs_{m,\tau}(\psi_1) \ge 1 - \epsilon_1$, requiring $\epsilon_1 \le 1/2\sqrt{N}$, gives $\erfs_{m,\tau}(\psi_1) + \sqrt{N}\epsilon_1 \ge 1/2$. Then, $\lnorm{\ket{f_{m, \tau}(\psi)} - \ket{x_1}}_2^2 \le 4 \epsilon_1 \sqrt{N}$. Setting $\epsilon_1 \le \epsilon^2/8\sqrt{N}$ ensures that $\lnorm{\ket{f_{m, \tau}(\psi)} - \ket{x_1}}_2 \le \epsilon/2$. Thus, the error bound holds when $m = \frac{\sqrt{2}}{\Delta}\log^{1/2}(\frac{16\sqrt{N}}{\pi \epsilon^2})$. 

We will now bound the complexity of preparing the state $\ket{\phi}$ such that $\lnorm{\ket{\phi}- \ket{f_{m,\tau}(\psi)}}_2 \le \epsilon / 2$.
We construct the state $\ket{\phi}$ by invoking \cref{theorem:generalized_end_to_end_complexity_for_f_0_is_0_for_efficiently_analytic_functions} transforming $\ket{\psi}$ by $f_{m,\tau}(x)$. Noting that $\tilde \gamma \le 2$, $\mathcal{N}\ge 1/\psi_1$, and that the degree of the approximation polynomial is given asymptotically by $O(m\log(1/\epsilon_0))$. Noting that \cref{theorem:generalized_end_to_end_complexity_for_f_0_is_0_for_efficiently_analytic_functions} gives the preceding $\epsilon/2$ error bound when $\epsilon_0 \le \frac{\epsilon\mathcal{N}^2}{2\gamma N}$, the degree of the approximating polynomial is thus $k \in O(\log(N/\epsilon\psi_1^2)\log^{1/2}(\sqrt{N}/\epsilon^2)/\Delta) = \tilde{O}(1/\Delta)$. This then gives the guarantee that $\lnorm{\ket{\phi} - \ket{f_{m, \tau}(\psi)}}_2 \le \epsilon/2$ with query complexity $\tilde{O}(\psi_1^{-1}\Delta^{-1})$ to controlled $U$ and $U^{\dagger}$ circuits, with $\tilde{O}(n\psi_1^{-1}\Delta^{-1})$ additional circuit depth. 

Thus, we have prepared a state $\ket{\phi}$ satisfying $\lnorm{\ket{\phi} - \ket{x_1}}_2 \le \epsilon$.
Since $\lnorm{\ket{\phi} - \ket{x_1}}_2 \le \epsilon$ implies $2(1 - \ip{x_1}{\phi}) \le \epsilon^2$, it is easy to show that the probability of successfully measuring $\ket{x_1}$, i.e., $\phi_1^2$, is lower-bounded as $\phi_1^2 \ge (1 - \tfrac{\epsilon^2}{2})^2$. Setting, e.g., $\epsilon=0.1$ ensures that the procedure success with probability at least $99/100$. 
\end{proof}

\section{Application: Continuous State Preparation}\label{section:continuous_state_prep}
In this section, we demonstrate that previous ideas on quantum state preparation can be easily captured in our framework of non-linear amplitude transformations, matching previous asymptotic complexities. Given some function $f$ (e.g. an efficiently computable function as described in \cite{rattew2022preparing}), and a set of $N$ points $\{x_j\}_j$ at which $f$ is to be evaluated, the goal is to prepare the $n$ qubit quantum state (with $N = 2^n$) $\ket{f}_n := \frac{1}{\mathcal{N}}\sum_{j=1}^N f(x_j)\ket{j}_n$ where $\mathcal{N}^2 = \sum_{j=1}^N |f(x_j)|^2$. In this section, we implicitly assume that the function is scaled so that its maximum over its domain is at most $1$.

There is a wide literature on this topic, originating with Grover's work on black-box state synthesis~\cite{grover2000synthesis}, in which a state with amplitudes following an arbitrary function (not just a discretized continuous function) can be prepared with complexity $O(\sqrt{N}/\mathcal{N})$, the inverse to a quantity termed the ``generalized filling ratio'' in ~\cite{rattew2022preparing}. A number of subsequent results have improved the practical efficiency of black-box techniques, for instance by reducing the cost of the required quantum arithmetic operations (see e.g.,~\cite{sanders2019black,bausch2022fast}).

When $f$ is continuous, a domain $[a, b]$ is additionally specified, and the set of points becomes a uniformly spaced grid of $N$ points $x_0, ..., x_{N-1}$ (with $x_0 = a$, $x_{N-1}=b$). In this setting, it was shown in \cite{rattew2022preparing} that the generalized filling ratio converges to an asymptotic constant they called the ``filling ratio'', $\mathcal{F}$. The rate of convergence was shown to be exponential in the number of qubits in the general case of Reimann integrable functions.
The filling ratio was shown to only depend on properties of the function, and intuitively can be visualized by a bound measuring the percentage of a bounding box filled by the function.

There are a number of other techniques in the state preparation literature with different approaches that in some cases exploit additional structure, e.g., via matrix product states (MPS)~\cite{holmes2020efficient,gonzalez2023efficient}, for efficiently integrable function~\cite{grover2002creating}, and a number of techniques both in the variational setting (e.g., ~\cite{lloyd2018quantum,zoufal2019quantum,romero2021variational}), and the in the distribution-specific setting (e.g.,~\cite{kitaev2008wavefunction,rattew2021efficient}).
However, of most relevance to our present work are the ideas in~\cite{mcardle2022quantum, gonzalez2023efficient}.
In~\cite{mcardle2022quantum}, an elegant quantum arithmetic free approach for preparing a continuous quantum state is presented.  Their algorithm creates an $N$-dimensional diagonal block encoding where the $j^{th}$ diagonal entry has value $\sin(j/N)$. To prepare quantum state with amplitudes proportional to a given function $f$ with an efficient uniform approximation, they then apply a polynomial approximately implementing $f((b -a)\arcsin(x) + a)$ to their block encoding, and then apply the result to the uniform superposition. The result is an algorithm which asymptotically prepares the desired state with overall complexity $\tilde{O}(\mathcal F^{-1})$ without using any quantum arithmetic. 
In \cite{gonzalez2023efficient}, they use MPS techniques to prepare a state with amplitudes proportional to the linear function. They then make the interesting observation that applying a non-linear function to the amplitudes of such a state effectively prepares a state with amplitudes proportional to the desired function.

Interestingly, when viewed through the lens of our diagonal block-encoding, \cite{gonzalez2023efficient} and \cite{mcardle2022quantum} are essentially two different ways of looking at the same idea (and essentially have the same asymptotic complexity). 
Viewing \cite{gonzalez2023efficient} through our framework, the state with amplitudes proportional to the linear function is converted into a diagonal block encoding $diag(0, 1/N, ..., (N-1)/N)$ (neglecting normalization of the state, which can be handled by rescaling the input to the function). The block-encoding is then transformed by the desired function, the block-encoding is applied to the uniform superposition, and the ancillary qubits are measured out. Clearly, this can be easily done with rigorous end-to-end complexity analysis in our framework through~\cref{theorem:generalized_end_to_end_complexity_for_previous_technique_for_efficiently_analytic_functions} and is expected to give an overall complexity of $\tilde{O}(\mathcal{F}^{-1})$.

From this perspective, \cite{gonzalez2023efficient} and \cite{mcardle2022quantum} are essentially the same, only \cite{mcardle2022quantum} prepares the diagonal block encoding of the linear function by first preparing a diagonal block encoding of $\sin(j/N)$ and then applying a polynomial approximation to $\arcsin(x)$, whereas \cite{gonzalez2023efficient} just directly prepares the desired block-encoding (albeit, without rigorous theoretical guarantees, due to the difficulty of such analysis with MPS techniques). 

To maintain our end-to-end error bounds while avoiding quantum arithmetic, in \cref{theorem:state_preparation_through_non_linear_transformation} we rederive the approach of \cite{mcardle2022quantum} in our framework, avoiding the need to prepare a quantum state with amplitudes proportional to the linear function without quantum arithmetic. We note that while our framework provides a conceptually simple way to rederive their results, our implementation requires slightly more ancillary qubits than theirs does.

Finally, we note that for functions satisfying $f(0)=0$, we may be able to get better than filling-ratio based complexity for state preparation (in certain settings) by instead using our importance sampling variant, but we leave further exploration of this idea for future work.

We will now provide a theorem stating the complexity of state preparation in the non-linear amplitude transformation framework, giving a simple sketch for the proof.

\begin{theorem}[State preparation through non-linear amplitude transformation~\cite{mcardle2022quantum,gonzalez2023efficient}]\label{theorem:state_preparation_through_non_linear_transformation}
Let $\epsilon > 0$.
Let $f : [a, b] \mapsto \mathbb R$ be a Lipschitz function such that $\max_{x\in[-1, 1]}|f(x)| \le 1$. Let $\{x_j\}_j$ be a uniform grid such that $x_0 = a$ and $x_{N-1} = b$.  
Define 
$\ket{f(\psi)}_n := \frac{1}{\mathcal{N}}\sum_j f(x_j)\ket{j}_n$, where $\mathcal{N} := \sum_j |f(x_j)|^2$.
Then, using~\cref{theorem:generalized_end_to_end_complexity_for_previous_technique_for_efficiently_analytic_functions}, a quantum state $\ket{\phi}_n$ satisfying $\lnorm{\ket{\psi}_n - \ket{f(\psi)}_n}_2 \le \epsilon$ can be prepared with overall complexity $\tilde{O}(\mathcal{F}^{-1})$, where $\mathcal{F}^{-1}$ is the function-dependent asymptotic constant called the ``filling ratio'' as per~\cite{rattew2022preparing}.
\end{theorem}
\begin{proof}
We now re-derive \cite{mcardle2022quantum} in our framework. 

Define the $n$-qubit state $\ket{\psi_0}_n := \frac{1}{\mathcal{N}_0}\sum_{j=0}^{N-1}\sin(j/N)\ket{j}_n$ where $N = 2^n$ and $\mathcal{N}_0^2 := \sum_{j=0}^{N-1}\sin(j/N)^2$.
There exists a unitary $U_{\psi_0}$ which is an $(1, 1, 0)$-SPBE for $\ket{\psi_0}$, where $\alpha = \sqrt{N}/\mathcal{N}_0$. Note that a unitary with the action 
\[U_{\psi_0}\ket{0}_{1 + n} = \frac{1}{\sqrt{N}}\left( \sum_{j=0}^N \left(\sin(j/N)\ket{0}_1 + \cos(j/N)\ket{1}_1 \right)\ket{j}_n
\right)\] is essentially a special case of the standard ``rotate'' gate (e.g., as used implicitly in \cite{harrow2009quantum}) only with Hadamards first applied.
As such, it can be implemented by first applying $(I_1 \otimes H^{\otimes n})$ (with $H$ the standard Hadamard gate), followed by a sequence of $n$ controlled $e^{-i Y \theta_j}$ gates (with $Y$ the standard Pauli-$Y$ gate) with the target qubit the ancilla, and the control of the $j^{th}$ gate set to the $j^{th}$ most-significant qubit of the main register (and with angle $2^{-(j+1)}$, counting from $1$ to $n$), finally followed by a Pauli-$X$ gate on the ancilla qubit. We can then rewrite the action of this gate as $U_{\psi_0}\ket{0}_{1+n} = (\frac{\mathcal{N}_0}{\sqrt{N}})\ket{0}_1\ket{\psi_0} + \sqrt{1 - (\frac{\mathcal{N}_0}{\sqrt{N}})^2}\ket{\bot}_{1+n}$ (for some orthogonal state $\ket{\bot}_{1+n}$, such that $\bra{\bot}_{1+n} \ket{0}_1\ket{\psi}_n = 0$). As such, by \cref{lemma:simple_spbe_construction}, we immediately find that $U_{\psi_0}$ is an $(\sqrt{N}/\mathcal{N}_0, 1, 0)$-SPBE for $\ket{\psi_0}_n$. 
Using observations first made in \cite{rattew2022preparing} about the asymptotic properties of continuous functions, it is clear that $\lim_{N\mapsto \infty}\sqrt{N}/\mathcal{N}_0 < 2$. Moreover, one can readily show that in the finite case, $\sqrt{N}/\mathcal{N}_0 \le 4$.

On the interval $x\in [-0.85, 0.85]$, as per Lemma 70 of \cite{gilyen2019quantum}, $\arcsin(x)$ is $\epsilon_1$-approximable by a polynomial $Q_l(x)$ with degree $l \in O(\log(1/\epsilon_1))$. 
By assumption, the input function $f$ is $\epsilon_0$-approximable by a polynomial $P_k(x)$ with degree $O(polylog(1/\epsilon_0))$ on the interval $[a, b]$. Consequently, \cref{lemma:composition_of_polynomial_approximations} says that the function $f((b-a)\arcsin(x) + a)$ is $\epsilon_0 + L\epsilon_1$ approximable by the polynomial $\tilde{P}_{m} := P_k((b - a)Q_l(x) + a)$ (where $m := kl$). Noting that $L$ is a constant, we can exponentially suppress $\epsilon_1$, meaning that effectively $m \in O(polylog(1/\epsilon_0))$. 

As a result, we can invoke \cref{theorem:generalized_end_to_end_complexity_for_previous_technique_for_efficiently_analytic_functions,generalizing_previous_results_to_spbe} with $\tilde{P}_{m}$ and our $(O(1), 1, 0)$-SPBE, obtaining the state $\ket{\phi}_n$ with total circuit complexity $\tilde{O}(\mathcal{F}^{-1})$, noting that $\sqrt{N}/\mathcal{N} = \mathcal{F}^{-1}$.
\end{proof}

\section{Conclusion}\label{section:discussion}
We have presented a new framework for applying non-linear functions to the amplitudes of quantum states. We give a diagonal block-encoding of the amplitudes of an input state specified by a state preparation unitary, simplifying analysis over previous approaches. We then give a new algorithm based off of importance sampling, which gives up to an exponential speedup in important applications. 
Our importance sampling algorithm also enables a new maximum finding application and result. Furthermore, we demonstrate how this framework can be used to efficiently apply a wide range of functions to the amplitudes of quantum states. We conclude by re-deriving state-preparation algorithms in our framework, showing its general versatility and power.

For future work, one direction is to generalize our results to work with quantum states with complex amplitudes, potentially using the recent Ref.~\cite{motlagh2023generalized}.
Additionally, for an $n$-qubit quantum state, it would be worth exploring ways to implement the diagonal block-encoding of state amplitudes with $O(1)$ ancilla qubits instead of the $O(n)$ currently required.
Another direction is to investigate the quantum state exponentiation setting \cite{lloyd2014quantum,kimmel2017hamiltonian} as opposed to the state preparation circuit input model.
Moreover, application to problems related to non-linear ordinary or partial differential equations are interesting to consider and may be enabled by our  subroutine. 
Furthermore, these techniques naturally lend themselves to further developing the theory of quantum machine learning, both in the parameterized quantum circuit model~\cite{cerezo2022challenges}, and in the fault-tolerant linear algebra model~\cite{gilyen2019quantum} (potentially requiring QRAM~\cite{jaques2023qram}).
Beyond the computational model considered in this work, certain quantum systems may offer non-linear functions at the hardware level, in contrast to the systematic polynomial approximation approach.

\section*{Acknowledgement}

We would like to thank B\'alint Koczor and Marco Pistoia for their helpful comments. 
We would like to thank Naixu Guo for input regarding his original paper~\cite{guo2021nonlinear} and for multiple subsequent discussions. 
We would also like to give particularly notable thanks to Shouvanik Chakrabarti for many helpful conversations, and in particular for helping to prove \cref{lemma:upper_bound_on_max_of_reduced_polynomial}. Finally, A.G.R. would like to thank the whole Global Technology Applied Research team at JPMorgan Chase for enabling this research.
This work is partially supported by the National Research Foundation, Singapore, and A*STAR under its CQT Bridging Grant and its Quantum Engineering Programme under grant NRF2021-QEP2-02-P05.

\section*{Disclaimer}
This paper was prepared for informational purposes with contributions from the Global Technology Applied Research center of JPMorgan Chase \& Co. This paper is not a product of the Research Department of JPMorgan Chase \& Co. or its affiliates. Neither JPMorgan Chase \& Co. nor any of its affiliates makes any explicit or implied representation or warranty and none of them accept any liability in connection with this paper, including, without limitation, with respect to the completeness, accuracy, or reliability of the information contained herein and the potential legal, compliance, tax, or accounting effects thereof. This document is not intended as investment research or investment advice, or as a recommendation, offer, or solicitation for the purchase or sale of any security, financial instrument, financial product or service, or to be used in any way for evaluating the merits of participating in any transaction.

\bibliography{bibliography}

\appendix
\section{Comparison to \cite{guo2021nonlinear}}
\label{appComparison}

Next, in \cref{theorem:modified_variant_of_previous_technique_complexity_of_applying_a_polynomial_to_a_state}, we derive an analogous guarantee to that in \cite{guo2021nonlinear}, with the same asymptotic properties, only using our modified diagonal block encoding framework. We also extend their result to $\ell_2$-norm error bound, which was not included in the previous work. 
The proof mirrors that in \cref{theorem:polynomial_transformation_of_real_state_amplitudes_via_importance_sampling}.
We subsequently derive a new result in \cref{theorem:generalized_end_to_end_complexity_for_previous_technique_for_efficiently_analytic_functions} (building on \cref{theorem:modified_variant_of_previous_technique_complexity_of_applying_a_polynomial_to_a_state}) with end-to-end complexity guarantees in the case where a non-polynomial function (with an efficient uniform approximation by a polynomial) is to be applied. The proof mirrors that in \cref{theorem:generalized_end_to_end_complexity_for_f_0_is_0_for_efficiently_analytic_functions}.

\begin{theorem}[Modification of non-linear transformation of real amplitudes from \cite{guo2021nonlinear} with $\ell_2$-error bounds]\label{theorem:modified_variant_of_previous_technique_complexity_of_applying_a_polynomial_to_a_state}
Given a state preparation unitary $U\ket{0}_n = \ket{\psi}_n = \sum_j \psi_j\ket{j}_n$ (such that $\forall j, \psi_j \in \mathbb{R}$), and a degree $k$ polynomial function, $P(x) = \sum_{j=0}^k a_j x^j$, such that $\gamma := \max_{x\in [-1, 1]}|P(x)|$, and $\forall j\in[k+1], a_j \in \mathbb{C}$, define $\mathcal{N}^2 := \sum_{j}|P(\psi_j)|^2$. Then, the technique in \cite{guo2021nonlinear} constructs an $\ell_2$-normalized quantum state $\ket{\phi}$ such that $\lnorm{\ket{\phi} - \frac{1}{\mathcal{N}}\sum_j f(\psi_j)\ket{j}}_2 \le \epsilon$ with arbitrarily high probability of success with a total of $O(k\gamma \sqrt{N}/\mathcal{N})$ queries to a controlled-$U$ circuit and a controlled-$U^{\dagger}$ circuit, with $O(kn\gamma\sqrt{N}/\mathcal{N})$ additional circuit depth, and with a total of $O(poly(k, \log(N\gamma^2/\epsilon^2 \mathcal{N}^2)))$ classical computation. This requires $O(n)$ ancillary qubits.
\end{theorem}
\begin{proof}

Let $f(x) := P(x)/4\gamma$, i.e. $\forall x, |f(x)| \le 1/4$.
Next, we invoke \cref{lemma:polynomial_transformation_of_the_real_part_block_encoding} to construct a $(1, n+4, \delta_0)$-block-encoding of $f(A)$, called $U_f$, with $O(kn)$ circuit depth, using a total of $k$ calls to a controlled $U_{\Re{A}}$ and $U_{\Re{A}}^{\dagger}$ circuit, and with $O(poly(k, \log(1/\delta_0)))$ classical time complexity. Then, in direct analogy to \cite{mitarai2019quantum, guo2021nonlinear}, we apply $U_f$ to the state $\ket{0}_{n+4}\ket{+^n}_n$ (where $\ket{+^n} := \frac{1}{\sqrt{N}}\sum_{j}\ket{j}$). Since our block-encoding has error $\delta_0$, we actually implement the transformation $\tilde{f}(A)$, where $\tilde{f}(A)$ is some linear operator satisfying $\lnorm{\tilde{f}(A) - f(A)}_2 \le \delta_0$. Define $\mathcal{N}_1 := \lnorm{f(A)\ket{+^n}}_2$ and $\mathcal{N}_2 := \lnorm{\tilde{f}(A)\ket{+^n}}_2$. Then, noting that $\frac{1}{\mathcal{N}}\sum_j f(\psi_j)\ket{j} = \frac{1}{\mathcal{N}_1}f(A)\ket{+^n}$, we wish to bound the total error $\epsilon$, in $\lnorm{\frac{1}{\mathcal{N}_1}f(A)\ket{+^n} - \frac{1}{\mathcal{N}_2}\tilde{f}(A)\ket{+^n}}_2 \le \epsilon$, and we also wish to bound the probability of success. 

Again, to bound the error, we must bound two quantities, $|\mathcal{N}_1 - \mathcal{N}_2|$ and $\lnorm{f(A)\ket{+^n} - \tilde{f}(A)\ket{+^n}}_2$, so that \cref{lemma:normalized_state_deviation_due_to_normalization} may be applied. The bound of,
\begin{align}\label{eqn:midstep_transform_of_real_state_amplitudes_unnormalized_bound_simple}
    \lnorm{f(A)\ket{+^n} - \tilde{f}(A)\ket{+^n}}_2 \le \delta_0
\end{align}
follows trivially from the fact that for any operator $M$ and for any vector $\bm{z}$, $\lnorm{M\bm{z}}_2\le \lnorm{M}_2\lnorm{\bm{z}}_2$, and the fact that $\lnorm{\tilde{f}(A) - f(A)}_2 \le \delta_0$.

Define the notation $\tilde{f}(A)\ket{+^n} = \frac{1}{\sqrt{N}}\sum_j\phi_j \ket{j}$. For any non-negative real numbers $a,b$, we have $|\sqrt{a} - \sqrt{b}| \le \sqrt{|a - b|}$. Consequently, $|\mathcal{N}_1 - \mathcal{N}_2| \le \sqrt{|\mathcal{N}_1^2 - \mathcal{N}_2^2|}$. Thus,
\begin{align}
    |\mathcal{N}_1 - \mathcal{N}_2| 
    &\le
    \sqrt{\left|
    \frac{1}{N}
    \sum_j |f(\psi_j)|^2 - |\phi_j|^2
    \right|}
    \le 
    \sqrt{\frac{1}{N}
    \max_j\left| |f(\psi_j)| - |\phi_j| \right|\sum_j\left(|f(\psi_j)| + |\phi_j| \right)}.
\end{align}
Since $||a| - |b|| \le |a - b|$, $\max_j\left| |f(\psi_j)| - |\phi_j| \right| \le \max_j\left| f(\psi_j) - \phi_j \right|$. From \cref{eqn:midstep_transform_of_real_state_amplitudes_unnormalized_bound_simple} we have that $\sum_j |f(\psi_j) - \phi_j|^2 \le \delta_0^2 N$, implying that $\forall j, |f(\psi_j) - \phi_j| \le \delta_0\sqrt{N}$. As a result,
\begin{align}
    |\mathcal{N}_1 - \mathcal{N}_2| 
    &\le
    \sqrt{\frac{\delta_0}{\sqrt{N}}\sum_j \left(|f(\psi_j)| + |\phi_j| \right)}.
\end{align}
Observing that $\lnorm{f(A)\ket{+^n} + \tilde{f}(A)\ket{+^n}}_2^2 \le 4$,  $\forall j, |f(\psi_j) + \phi_j| \le 2\sqrt{N}$. Thus,
\begin{align}\label{eqn:midstep_transform_of_real_state_amplitudes_norm_dif_bound_simple}
    |\mathcal{N}_1 - \mathcal{N}_2| 
    &\le
    \sqrt{2\delta_0 N}.
\end{align}
Combining 
\cref{eqn:midstep_transform_of_real_state_amplitudes_unnormalized_bound_simple}, \cref{eqn:midstep_transform_of_real_state_amplitudes_norm_dif_bound_simple} and \cref{lemma:normalized_state_deviation_due_to_normalization}, we then obtain the following bound on the total $\ell_2$-norm error,
\begin{align}
    \lnorm{\frac{1}{\mathcal{N}_1}f(A)\ket{+^n} - \frac{1}{\mathcal{N}_2}\tilde{f}(A)\ket{+^n}}_2 \le
    \frac{\sqrt{2\delta_0 N} + \delta_0}{\mathcal{N}_1}
    \le 
    \frac{16N\gamma\sqrt{\delta_0}}{\mathcal{N}},
\end{align}
using the fact that $\mathcal{N}_1 = \mathcal{N}/4\gamma\sqrt{N}$.
We must now lower-bound the probability of success, and thus upper-bound the number of amplitude amplification steps required to obtain an arbitrarily high probability of success. 
The probability of success of the procedure is given by $\lnorm{\tilde{f}(A)\ket{+^n}}_2^2 = \mathcal{N}_2^2$. From  $|\mathcal{N}_1^2 - \mathcal{N}_2^2| \le 2\delta_0N$, follows that
$\mathcal{N}_2^2 \ge \mathcal{N}_1^2 - 2\delta_0 N$.
Imposing the restriction that $\mathcal{N}_1^2 \ge 2\delta_0 N$, which as we show later does not result in a loss of generality, this gives a number of AA steps required of $O(1/\sqrt{\mathcal{N}_1^2 - 2\delta_0 N})$. 
To obtain an overall error of at most $\epsilon$,  
set $\delta_0 = \epsilon^2\mathcal{N}^2/256 \gamma^2 N^2$, satisfying the requirement that $\mathcal{N}^2/N \ge 32\delta_0 N$, since $1 \ge \mathcal{N}/\gamma\sqrt{N}$.
This results in the upper-bound on the number of AA steps required of $O(\gamma\sqrt{N}/\mathcal{N})$. Thus, since we require $O(\gamma\sqrt{N}/\mathcal{N})$ queries to the block-encoding of $f(A)$, the overall algorithm succeeds with arbitrarily high probability of success with a total of $O(k \gamma\sqrt{N}/\mathcal{N})$ queries to a controlled $U$ and controlled $U^{\dagger}$ circuit, with a total of $O(k n\gamma \sqrt{N}/\mathcal{N})$ circuit depth, and with $O(poly(k, \log(N\gamma^2/\epsilon^2\mathcal{N}^2)))$ classical time complexity. 
\end{proof}

\begin{theorem}\label{theorem:generalized_end_to_end_complexity_for_previous_technique_for_efficiently_analytic_functions}
We are given an $n$-qubit circuit $U$ such that $U\ket{0} = \sum_{j=1}^N\psi_j\ket{j}$, with $N = 2^n$ and $\forall j, \psi_j\in \mathbb R$, and a function $f(x)$ with $f(0) > 0$, $\gamma:=\max_{x\in[-1,1]}|f(x)|$, and $\mathcal{N}^2 := \sum_{j=1}^N |f(\psi_j)|^2$. 
For $\epsilon >0$, let $f(x)$ be $\frac{\gamma N}{\epsilon\mathcal{N}^2}$-approximable by a polynomial $P_k(x)$ with degree 
$k = k(\gamma N/\epsilon\mathcal{N}^2)$, where $k$ is a function describing the degree of the polynomial in terms of the error.
Additionally, define $\gamma := \max_{x \in [-1, 1]}|P_k(x)/x|$. 
Then, the algorithm of \cite{guo2021nonlinear}, as described in \cref{theorem:modified_variant_of_previous_technique_complexity_of_applying_a_polynomial_to_a_state}, can prepare a state $\epsilon$ close in $\ell_2$ norm-distance (for small $\epsilon$) from $\frac{1}{\mathcal{N}}\sum_j f(\psi_j)\ket{j}$ (where $\mathcal{N}$ is the normalization factor) with arbitrarily high probability of success with $O(\log(\gamma N/\epsilon\mathcal{N}^2)\gamma \sqrt{N}/\mathcal{N})$ query complexity to a controlled $U$ and $U^{\dagger}$ circuit, with overall circuit depth of $O(n\log(\gamma N/\epsilon\mathcal{N}^2)\gamma \sqrt{N}/\mathcal{N})$, $O(polylog(N\gamma^2/\epsilon^2\mathcal{N}^2))$ classical computation, and $O(n)$ ancillary qubits.
\end{theorem}
\begin{proof}
As in \cref{theorem:generalized_end_to_end_complexity_for_f_0_is_0_for_efficiently_analytic_functions}, we bound the $\ell_2$-norm error from the polynomial approximation, and then the $\ell_2$-norm error from the algorithmic error, combining the two via triangle inequality. 
\end{proof}

\section{State Preparation Block Encodings}\label{section:state_preparation_block_encodings}

So far, we have used the language of state preparation unitaries, as per \cref{def:state_prep_unitary_input}. However, states are often produced utilizing ancillary qubits or projective partial measurements which are not naturally captured in this framework. Consequently, we introduce the notion of a state preparation block encoding (SPBE) in \cref{def:state_prep_block_encoding} to generalize our previous results. SPBEs are a special case of block-encodings (as proposed by \cite{gilyen2019quantum}, building off of the previous work of e.g., \cite{low2017hamiltonian,chakraborty2018power,van2020quantum}) only imposing that the $a+n$ qubit unitary block-encoding is applied to the $\ket{0}_{a + n}$ state. 
The results in this section will be obvious to authors familiar with the QSVT and QSP literature, but we include these results for the sake of generality.

In \cref{lemma:simple_spbe_construction}, we provide a simple and common example of an SPBE, which is commonly encountered in the literature. 
In \cref{lemma:spbe_fixed_amplitude_amplification}, we show how an SPBE can be trivially converted to a state preparation unitary in a higher-dimensional space through fixed-point amplitude amplification.
In \cref{lemma:deviation_in_output_state_of_function_with_ddeviation_in_input_states} we prove that amplitude transformation procedures (which can be thought of as mappings between SPBEs) are not sensitive to perturbations in their input states. 
We then provide a simple observation in \cref{generalizing_previous_results_to_spbe} which shows how \cref{theorem:polynomial_transformation_of_real_state_amplitudes_via_importance_sampling,theorem:modified_variant_of_previous_technique_complexity_of_applying_a_polynomial_to_a_state,theorem:generalized_end_to_end_complexity_for_f_0_is_0_for_efficiently_analytic_functions,theorem:generalized_end_to_end_complexity_for_previous_technique_for_efficiently_analytic_functions} can be used in this different input model essentially without modification.

\begin{definition}[State Preparation Block Encoding (SPBE)]\label{def:state_prep_block_encoding}
Let $\alpha \ge 1$, $a\in \mathbb{N}$ and $\epsilon \ge 0$.
We call the $n + a$ qubit unitary matrix $U_{\psi}$ an $(\alpha, a, \epsilon)$-SPBE for the $n$ qubit quantum state $\ket{\psi}_n$, if 
\begin{align}
    \lnorm{\ket{\psi}_n - \alpha(\bra{0}_a\otimes I_n)U_{\psi}\ket{0}_{a + n}}_2 \le \epsilon.
\end{align}
\end{definition}

\begin{remark}
    A $(1, 0, 0)$-SPBE is simply a state preparation unitary, as per \cref{def:state_prep_unitary_input}.
\end{remark}

\begin{lemma}[A standard SPBE]\label{lemma:simple_spbe_construction}
    Define the $n$ qubit quantum state $\ket{\psi}_n$, and the $a + n$ qubit state $\ket{\bot}_{a+n}$, such that $\bra{\bot}_{a+n}(\ket{0}_a\otimes I_n) = \bm 0$ (i.e., the all zero vector). Then, any $n+a$ qubit unitary satisfying
    \begin{align}
        U_{\psi}\ket{0}_{a + n} = (1/\alpha)\ket{0}_a\ket{\psi}_n + \sqrt{1 - |1/\alpha|^2}\ket{\bot}_{a + n},
    \end{align}
    is an $(\alpha, a, 0)$-SPBE for the $n$-qubit quantum state $\ket{\psi}_n$.
\end{lemma}
\begin{proof}
    This follows immediately from \cref{def:state_prep_block_encoding}.
\end{proof}

\begin{lemma}\label{lemma:spbe_fixed_amplitude_amplification}
Let $\epsilon_0 \in [0, 1/2]$, $\alpha \ge 1$, $a \in \mathbb N$.
Given a unitary $U_{\psi}$, an $(\alpha, a, \epsilon_0)$-SPBE for the $\ell_2$-normalized quantum state $\ket{\psi}$ with circuit complexity $O(T)$, we can construct a $(1, a + 1, \sqrt{2\epsilon_0} + \epsilon_1)$-SPBE with circuit complexity $O(T(a + n)\alpha \log(1/\epsilon_1))$ and with $O(\alpha\log(1/\epsilon_1))$ queries to a $U_{\psi}$ and $U_{\psi}^{\dagger}$ circuit.
\end{lemma}
\begin{proof}
    This follows similarly to the fixed-point amplitude amplification result of \cite{gilyen2019quantum}, which is an improved version of~\cite{grover2005different} (the later, clearly not having been cast in the block-encoding framework). We give a simple sketch of the procedure, noting that it is possible to implement this with fewer ancilla qubits than we show.

    Define the projector $\Pi := \op{0}{0}_a\otimes I_n$. 
    Define the $a + n$ qubit operator $A:= \Pi U_{\psi}\op{0}{0}_{a + n}$. 
    Define the (not necessarily $\ell_2$-normalized) vector $\ket{\phi} := (\bra{0}_{a}\otimes I_n) U_{\psi}\ket{0}_{a + n}$. Then, it is easy to show that $A = (\ket{0}_a\ket{\phi}_n)\bra{0}_{a + n}$.

    We can construct a $(1, 1, 0)$-block-encoding of $\Pi$, by applying a not gate to the single ancilla of the new block-encoding, and then applying a multiple-controlled $X$ gate targeting the ancilla, with the $a$ controls conditioned on the $\ket{0}_a\otimes I_n$ state of the main register. We can similarly construct a $(1, 1, 0)$-block-encoding of $\op{0}{0}_{a+n}$ with the same circuit, only with the $a + n$ controls conditioned on the $\ket{0}_{a + n}$ state of the main register instead. Noting that $U_{\psi}$ is a $(1, 0, 0)$-block-encoding of itself, we can then take the product of the three preceding block encodings in the appropriate order (via \cref{lemma:product_of_block_encodings}) to construct the unitary $U_A$, a $(1, 2, 0)$-block-encoding of $A$.

    Let $\mathcal{N}_{\phi} := \lnorm{\ket{\phi}}_2$, and define $\ket{\tilde{\phi}} := \ket{\phi}/\mathcal{N}_{\phi}$. Clearly, $A$ is rank-one, and has non-zero singular value $\lnorm{A}_2 = \mathcal{N}_{\phi}$. We need to lower-bound this value, so that we know what threshold to set in our application of the sign function, to map this singular value to one. By definition of an SPBE, we know that $\lnorm{\ket{\psi}_n - \alpha \ket{\phi}_n}_2 \le \epsilon_0$, which by reverse-triangle inequality implies that $\mathcal{N}_{\phi} \ge (1 - \epsilon_0)/\alpha \ge 1/(2\alpha)$ (additionally using the assumption that $\epsilon_0 \le 1/2$).

    Following e.g., Theorems 26 and 27 of \cite{gilyen2019quantum}, 
    we can construct an odd polynomial of degree $k \in O(\alpha\log(1/\epsilon_1))$ approximating the sign function on the interval 
    $[-1, 1]\backslash(-1/\alpha, 1/\alpha)$ which can be applied to our block encoding to effectively bring its singular value $\epsilon_1$-close to $1$,    
    using an additional ancillary qubit and $k$ calls to $U_A$ and $U_A^{\dagger}$. This yields the unitary $U_{\tilde{A}}$, a $(1, 3, \epsilon_1)$-block-encoding of $\tilde{A} := (\ket{0}_a\ket{\tilde{\phi}}_n)\bra{0}_{a + n}$ such that $\lnorm{\tilde{A} - (\bra{0}_3\otimes I_{a + n}) U_{\tilde{A}}(\ket{0}_3\otimes I_{a + n})}_2 \le \epsilon_1$.
    Let $B := (\bra{0}_3\otimes I_{a + n}) U_{\tilde{A}}(\ket{0}_3\otimes I_{a + n})$, and define $E:= \tilde{A} - B$, i.e. $\lnorm{E}_2 \le \epsilon_1$.
    Noting that $U_A$ asymptotically consists a call to $U_{\psi}$ and an $a+n$ controlled $X$ gate, and that an $a + n$ controlled-Toffoli can be decomposed into single and two qubit gates with depth $O(a + n)$ (see e.g., \cite{saeedi2013linear}), $U_{\tilde{A}}$ has circuit complexity $O(T(a + n)\alpha\log(1/\epsilon_1))$.

    We will now show that $U_{\tilde{A}}$ is a $(1, a + 3, \sqrt{2\epsilon_0} + \epsilon_1)$-SPBE for $\ket{\psi}$. To determine what SPBE we have, we must bound $\lnorm{\ket{\psi}_n - (\bra{0}_{3+a}\otimes I_n)U_{\tilde{A}}\ket{0}_{3 + a + n}}_2$. Noting that $\ket{0}_{3 + a + n} = (\ket{0}_3\otimes I_{n+a})(\ket{0}_{a + n})$ and that $\bra{0}_{3+a}\otimes I_n = (\bra{0}_a\otimes I_n)(\bra{0}_3\otimes I_{a + n})$, $(\bra{0}_{3+a}\otimes I_n)U_{\tilde{A}}\ket{0}_{3 + a + n} = (\bra{0}_a\otimes I_n) B \ket{0}_{a + n}$. Consequently, since $\tilde{A} = (\ket{0}_a\ket{\tilde{\phi}}_n)\bra{0}_{a + n}$, and inserting $B = \tilde{A} - E$,
    \begin{align}
        \lnorm{\ket{\psi}_n - (\bra{0}_a\otimes I_n) B \ket{0}_{a + n}}_2
        \le 
        \lnorm{\ket{\psi}_n - \ket{\tilde{\phi}}_n}_2 + \lnorm{E}_2.
    \end{align}
    Noting that applying the reverse triangle inequality to $\lnorm{\ket{\psi}-\alpha \ket{\phi}}_2 \le \epsilon_1$ implies that $\alpha\mathcal{N}_{\phi} \le 1 + \epsilon_0$, $1 - \epsilon_0 \le \alpha\mathcal{N}_{\phi}$, and $\frac{1 + \alpha^2\mathcal{N}_{\phi}^2 - \epsilon_0^2}{2\alpha} \le \Re{\ip{\psi}{\phi}}$, assuming that $\epsilon_0 \le 1$, it can be easily shown that $\lnorm{\ket{\psi} - \ket{\tilde{\phi}}}_2 \le \sqrt{2\epsilon_0}$. As a result, we get the final bound, $\lnorm{\ket{\psi}_n - (\bra{0}_a\otimes I_n) B \ket{0}_{a + n}}_2 \le \sqrt{2\epsilon_0} + \epsilon_1$, and so $U_{\tilde{A}}$ is a $(1, a + 3, \sqrt{2\epsilon_0} + \epsilon_1)$-SPBE for $\ket{\psi}_n$.

    The circuit and query complexities follow from simple bookkeeping. 
\end{proof}

\begin{lemma}\label{lemma:deviation_in_output_state_of_function_with_ddeviation_in_input_states}

Define the $N$-dimensional quantum states $\ket{\psi} := \sum_{j}\psi_j\ket{j}$, and $\ket{\phi} := \sum_{j}\phi_j\ket{j}$, where both states are $\ell_2$-normalized. Let $f : \mathbb{C} \mapsto \mathbb{C}$ be an $L$-Lipschitz continuous function. Define $\gamma := \max_{x\in[-1, 1]}|f(x)|$. Define the unnormalized quantum states $\ket{f(\psi)} := \sum_{j}f(\psi_j)\ket{j}$, $\ket{f(\phi)} := \sum_{j}f(\phi_j)\ket{j}$, and normalization constants $\mathcal{N}_{\psi} := \lnorm{\ket{f(\psi)}}_2$ and $\mathcal{N}_{\phi} := \lnorm{\ket{f(\phi)}}_2$.
Then, if $\lnorm{\ket{\phi} - \ket{\psi}}_2 \le \epsilon_0$,
\begin{align}
    \lnorm{\frac{1}{\mathcal{N}_{\psi}}\ket{f(\psi)} - \frac{1}{\mathcal{N}_{\phi}}\ket{f(\phi)}}_2 \le 3\gamma L \epsilon_0 N/\mathcal{N}_{\psi}^2.
\end{align}
\end{lemma}
\begin{proof}
Using the fact that for any vector $\bm{x}$, $\lnorm{\bm x}_{\infty} \le \lnorm{\bm{x}}_2$, we immediately get that $\max_j |\phi_j - \psi_j| \le \epsilon_0$ as well. 
Denote the corresponding $\ell_2$-normalized quantum states as $\ket{\tilde{f}(\psi)} := \ket{f(\psi)}/\mathcal{N}_{\psi}$ and $\ket{\tilde{f}(\phi)} := \ket{f(\phi)}/\mathcal{N}_{\phi}$.
We will now prove upper-bounds on $|\mathcal{N}_{\psi}-\mathcal{N}_{\phi}|$ and $\lnorm{\ket{f(\psi)} - \ket{f(\phi)}}_2$ and a lower-bound on $\mathcal{N}_{\phi}$ so that we may apply \cref{lemma:normalized_state_deviation_due_to_normalization} to bound $\lnorm{\ket{\tilde{f}(\psi)} - \ket{\tilde{f}(\phi)}}_2$.
Since $f$ has a Lipschitz constant of $L$, and $|\psi_j - \phi_j| \le \epsilon_0$, $|f(\psi_j) - f(\phi_j)| \le L\epsilon_0$. Consequently, by simple algebra with the square of the norm of the difference, we get $\lnorm{\ket{f(\psi)} - \ket{f(\phi)}}_2 \le L\epsilon_0 \sqrt{N}$. 
Moreover, $\frac{\gamma \sqrt{N}}{\mathcal{N}_{\psi}}\ge 1 \implies \epsilon_0\gamma N/\mathcal{N}_{\psi} \ge \epsilon_0 \sqrt{N}$, and so  $\lnorm{\ket{f(\psi)} - \ket{f(\phi)}}_2 \le L\epsilon_0\gamma N/\mathcal{N}_{\psi}$.
Additionally, since $|\mathcal{N}_{\psi}^2-\mathcal{N}_{\phi}^2| \le 2\gamma L \epsilon_0 N$, then 
$|\mathcal{N}_{\psi}-\mathcal{N}_{\phi}| = \frac{|\mathcal{N}_{\psi}^2-\mathcal{N}_{\phi}^2|}{|\mathcal{N}_{\psi}+\mathcal{N}_{\phi}|} \le |\mathcal{N}_{\psi}^2-\mathcal{N}_{\phi}^2|/\mathcal{N}_{\psi} \le (2\gamma L \epsilon_0 N)/\mathcal{N}_{\psi}$. 
Using $\max\{\mathcal{N}_{\psi}, \mathcal{N}_{\phi}\} \ge \mathcal{N}_{\psi}$, we can then invoke \cref{lemma:normalized_state_deviation_due_to_normalization}, getting the bound $\lnorm{\ket{\tilde{f}(\psi)} - \ket{\tilde{f}(\phi)}}_2 \le 3\gamma L \epsilon_0 N/\mathcal{N}_{\psi}^2$. To get the overall error-bound of $\epsilon$, we can set $\epsilon_0 \le \frac{\mathcal{N}_{\psi}^2 \epsilon}{3 L\gamma N}$.
\end{proof}

\begin{definition}[State amplitude transformation procedure]\label{def:state_amplitude_transformation_procedure}
A state amplitude transformation procedure (as in e.g., \cref{theorem:polynomial_transformation_of_real_state_amplitudes_via_importance_sampling,theorem:modified_variant_of_previous_technique_complexity_of_applying_a_polynomial_to_a_state,theorem:generalized_end_to_end_complexity_for_f_0_is_0_for_efficiently_analytic_functions,theorem:generalized_end_to_end_complexity_for_previous_technique_for_efficiently_analytic_functions}) accepts a state preparation unitary $U$ such that $U\ket{0}_n = \sum_j \psi_j\ket{j}_n =: \ket{\psi}_n$, and an $L$-Lipschitz continuous function $f : \mathbb{C} \mapsto \mathbb{C}$. Define $\ket{\tilde{f}(\psi)}_n := \frac{1}{\mathcal{N}}\sum_{j}f(\psi_j)\ket{j}_n$ (where $\mathcal{N}^2 := \sum_{j}|f(\psi_j)|^2$).
It then produces a quantum state $\ket{\phi}$, satisfying $\lnorm{\ket{\tilde{f}(\psi)}_n - \ket{\phi}_n}_2 \le \epsilon$ with $O(T_{n, \epsilon})$ query complexity to a controlled $U$ and $U^{\dagger}$ circuit, and with $\tilde{O}(T_{n, \epsilon})$ circuit complexity. The subscripts make the error-dependence and dimension-dependence explicit.     
\end{definition}

Note that for an $(\alpha, a, \epsilon)$-SPBE, $\alpha$ can be intuitively thought of as an estimate proportional to the amount of amplitude amplification required to bring the norm of the relevant subspace of the state preparation block encoding to near unity.
In the following theorem, we show that up to logarithmic factors (and a multiplicative $\alpha$ term), state amplitude transformation procedures operating on state preparation unitaries and state amplitude transformation procedures operating state preparation block encodings have equivalent complexity.

\begin{theorem}\label{generalizing_previous_results_to_spbe}
We are given an $L$-Lipschitz continuous function $f : \mathbb{C} \mapsto \mathbb{C}$. Define $\gamma := \max_{x\in[-1, 1]}|f(x)|$.
Given a state amplitude transformation procedure (as per \cref{def:state_amplitude_transformation_procedure}) which accepts an $n$ qubit state preparation unitary $U$ such that $U\ket{0}_n = \sum_j \psi_j\ket{j}_n$ and outputs an $\ell_2$-normalized quantum state $\epsilon$-close (in $\ell_2$-norm distance) to $\ket{\tilde{f}(\psi)}_n := \frac{1}{\mathcal{N}}\sum_{j}f(\psi_j)\ket{j}_n$ (where $\mathcal{N}^2 := \sum_{j}|f(\psi_j)|^2$) with $O(T_{n, \epsilon})$ queries to a controlled $U$ and $U^{\dagger}$ circuit and with $\tilde{O}(T_{n, \epsilon})$ circuit complexity, then there exists a procedure which accepts a unitary $U_{\psi}$, an $(\alpha, a', 0)$-SPBE for $\ket{\psi}_n$ with circuit complexity $T_{\psi}$, and outputs a state $\epsilon$-close to $\ket{\tilde{f}(\psi)}_n$ with $O(\alpha T_{a+ n, \epsilon} \log(\frac{\gamma^2 L^2 N^2}{\epsilon^2 \mathcal{N}^4}))$ query complexity to a controlled $U_{\psi}$ and $U_{\psi}^{\dagger}$ circuit, and with $\tilde{O}(\alpha T_{a' + n, \epsilon} \log(\frac{\gamma^2 L^2 N^2}{\epsilon^2 \mathcal{N}^4}))$ total additional circuit depth. 

\end{theorem}
\begin{proof}

Let $a := a' + 3$, and $\epsilon_0, \epsilon_0' \ge 0$. Let $\epsilon_0 = \sqrt{2\epsilon_0'}$.
Using \cref{lemma:spbe_fixed_amplitude_amplification}, we can convert $U_{\psi}$, an $(\alpha, a', \epsilon_0')$-SPBE for $\ket{\psi}_n$, into $U_A$, a $(1, a, \epsilon_0 + \epsilon_1)$-SPBE for $\ket{\psi}_n$ with circuit complexity $O(T(a + n)\alpha \log(1/\epsilon_1))$ and with $O(\alpha\log(1/\epsilon_1))$ queries to $U_{\psi}$.

Define the set of $2^a$ orthogonal projectors, $\{p_j\}_j$ with $p_j := \op{j}{j}_a\otimes I_n$. Then,
\begin{align}
    \lnorm{\ket{0}_a\ket{\psi}_n - U_A\ket{0}_{a + n}}_2^2
    =
    \lnorm{\ket{0}_a\ket{\psi}_n - p_0U_A\ket{0}_{a + n}}_2^2
    +
    \sum_{j=1}^{2^a - 1}\lnorm{p_jU_A\ket{0}_{a + n}}_2^2.
\end{align}
We will now use the definition of $U_A$ as an SPBE twice, i.e. that $\lnorm{\ket{\psi}_n - (\bra{0}_a\otimes I_n)U_A\ket{0}_{a+n}}_2 \le \epsilon_0 + \epsilon_1$.
First, $\lnorm{\ket{0}_a\ket{\psi}_n - p_0\ket{0}_{a + n}}_2 \le \epsilon_0 + \epsilon_1$ trivially following by applying the definition, and by factoring out a common $\ket{0}_a\otimes I_n$ term. 
Next, by using the definition, and by using the reverse triangle inequality, it can be easily shown that $\lnorm{p_0 U_A \ket{0}_{a+n}}_2^2 \ge (1 - \epsilon_0 - \epsilon_1)^2$.
Similarly, by observing that $1 - \lnorm{p_0 U_A \ket{0}_{a+n}}_2^2 = \sum_{j=1}^{2^a - 1}\lnorm{p_jU_A\ket{0}_{a + n}}_2^2$, it is easy to prove $\sum_{j=1}^{2^a - 1}\lnorm{p_jU_A\ket{0}_{a + n}}_2^2 \le 1 - (1 - \epsilon_0 - \epsilon_1)^2$. Consequently, we get the overall bound $\lnorm{\ket{0}_a\ket{\psi}_n - U_A\ket{0}_{a + n}}_2 \le \sqrt{2(\epsilon_0 + \epsilon_1)} =: \delta$.

If the function $f$ is applied to the state $\ket{0}_{a}\ket{\psi}_n$ without error, the state $\ket{0}_a \ket{\tilde{f}(\psi)}_n$ is obtained. Our ultimate goal is to bound the error and associated complexity in obtaining the state $\ket{\phi}_{a + n}$ such that $\lnorm{\ket{0}_a \ket{\tilde{f}(\psi)}_n - \ket{\phi}_{a + n}}_2 \le \epsilon$.
To do this, note that our procedure has two primary sources of error: error in the implementation of the transformation corresponding to $f$, and error in the state input to the transformation procedure. 
Let $\ket{\lambda}_{a + n}$ correspond to the state produced by a procedure \textit{exactly} applying $f$ to $U_A\ket{0}_{a + n}$. Then, $\lnorm{\ket{0}_a \ket{\tilde{f}(\psi)}_n - \ket{\phi}_{a + n}}_2 \le \lnorm{\ket{0}_a \ket{\tilde{f}(\psi)}_n - \ket{\lambda}_{a + n}}_2 + \lnorm{\ket{\lambda}_{a + n} - \ket{\phi}_{a + n}}_2$.

Let $\epsilon_2 > 0$.
By assumption of a state amplitude transformation procedure, for any given input state (in this case the input is $U_A\ket{0}_{a + n}$), the $\ell_2$-norm error in the output is bounded by $\epsilon_2$ with query complexity $O(T_{a + n, \epsilon_2})$. I.e. $\lnorm{\ket{\lambda}_{a + n} - \ket{\phi}_{a + n}}_2 \le \epsilon_2$, with query complexity to a controlled $U_A$ and $U_A^{\dagger}$ circuit of $O(T_{a + n, \epsilon_2})$. Similarly, since we have just shown that $\lnorm{\ket{0}_a\ket{\psi}_n - U_A\ket{0}_{a + n}}_2 \le \delta$, from \cref{lemma:deviation_in_output_state_of_function_with_ddeviation_in_input_states} it follows that $\lnorm{\ket{0}_a \ket{\tilde{f}(\psi)}_n - \ket{\lambda}_{a + n}}_2 \le 3\gamma L \delta N /\mathcal{N}^2$. 

Splitting the error equally among the two sources, and noting that $\delta = \sqrt{2(\epsilon_0 + \epsilon_1)}$, imposing the following conditions ensures that the overall error is bounded by $\epsilon$: $\epsilon_2 \le \epsilon/2$, $\epsilon_0 \le \epsilon_1 \le \frac{\epsilon^2 \mathcal{N}^4}{144 \gamma^2 L^2 N^2}$. If $\epsilon_0 = 0$, we can then set the bound $\epsilon_1 \le \frac{\epsilon^2 \mathcal{N}^4}{72 \gamma^2 L^2 N^2}$.

Noting that $U_A$ has a query complexity to $U_{\psi}$ of $O(\alpha \log(1/\epsilon_0))$, and that $U_A$ is queried a total of $O(T_{a + n, \epsilon_2})$ times, the overall query complexity to $U_{\psi}$ is given by $O(\alpha T_{a+ n, \epsilon} \log(\frac{\gamma^2 L^2 N^2}{\epsilon^2 \mathcal{N}^4}))$. The circuit complexity follows similarly. 
\end{proof}

\end{document}